\documentclass[10pt,twocolumn,twoside]{IEEEtran} 
\usepackage{authblk}

\usepackage{epsfig}
\usepackage{booktabs}
\usepackage{graphicx}
\usepackage[format=hang,font={sf,scriptsize},labelfont=bf]{subfig}
\usepackage[font=small,skip=0pt]{caption}
\usepackage{cleveref}
\usepackage{chngcntr}

\usepackage{diagbox}
\usepackage{color}

\usepackage[thmmarks]{ntheorem}
{
	\theoremstyle{nonumberplain}	
	\theoremheaderfont{\bfseries}
	\theorembodyfont{\normalfont}
	\theoremsymbol{\mbox{$\Box$}}
	\newtheorem{proof}{Proof.}
	
}
\newtheorem{thm}{Theorem}
\newtheorem{prop}{Proposition}

\usepackage{enumitem}
\newlist{myEnumerate}{enumerate}{2}
\setlist[myEnumerate,1]{leftmargin=*,topsep=0pt,itemsep=0pt,parsep=0pt,label=\normalfont\textbf{\arabic*)}}
\setlist[myEnumerate,2]{leftmargin=*,topsep=0pt,itemsep=0pt,parsep=0pt,label=\normalfont(\alph*)}

\usepackage{algorithm,algorithmic}
\newlength\myindent
\usepackage{xspace}
\newcommand*{\eg}{\textit{e.g.}\@\xspace}

\usepackage{amsmath,amsfonts,amssymb,bm,yhmath,accents}
\usepackage{mathtools}
\newcommand\numberthis{\addtocounter{equation}{1}\tag{\theequation}}

\usepackage[pagebackref=true,breaklinks=true,letterpaper=true,colorlinks,bookmarks=false]{hyperref}

\begin{document}

\title{Dynamic Parallel and Distributed Graph Cuts}


\author[1,2]{Miao Yu}
\author[1,4]{Shuhan Shen}
\author[1,3,4]{Zhanyi Hu$^{\star}$}
\affil[1]{National Laboratory of Pattern Recognition, Institute of Automation, Chinese Academy of Sciences, Beijing 100190, China}
\affil[2]{Zhongyuan University of Technology, Zhengzhou 450007, China}
\affil[3]{CAS Center for Excellence in Brain Science and Intelligence Technology, Chinese Academy of Sciences, Beijing 100190, China}
\affil[4]{University of Chinese Academy of Sciences, Beijing 100049, China}
\affil[ ]{\textit {\{myu, shshen, huzy\}@nlpr.ia.ac.cn}}
\maketitle


\begin{abstract}
    Graph cuts are widely used in computer vision. To speed up the optimization process and improve the scalability for large graphs, Strandmark and Kahl~\cite{strandmark2010parallel,strandmark2011parallel} introduced a splitting method to split a graph into multiple subgraphs for parallel computation in both shared and distributed memory models. However, this parallel algorithm (the parallel BK-algorithm) does not have a polynomial bound on the number of iterations and is found to be non-convergent in some cases~\cite{shekhovtsov-12-mixed_maxflow} due to the possible multiple optimal solutions of its sub-problems.

    To remedy this non-convergence problem, in this work, we first introduce a merging method capable of merging any number of those adjacent sub-graphs that can hardly reach agreement on their overlapping regions in the parallel BK-algorithm. Based on the pseudo-boolean representations of graph cuts, our merging method is shown to be effectively reuse all the computed flows in these sub-graphs. Through both splitting and merging, we further propose a dynamic parallel and distributed graph cuts algorithm with guaranteed convergence to the globally optimal solutions within a predefined number of iterations. In essence, this work provides a general framework to allow more sophisticated splitting and merging strategies to be employed to further boost performance. Our dynamic parallel algorithm is validated with extensive experimental results.
\end{abstract}


\section{Introduction}
Graph cuts optimization plays an important role in solving the following energy minimization problem, which is usually derived from a Maximum A \emph{Posteriori} (MAP) estimation of a Markov Random Field (MRF)~\cite{blake2011markov}:
\begin{equation}
    E(\mathbf{X}) = \sum_{c\in\mathcal{C}}\psi_c(\mathbf{x}_c),
    \label{eq:energyDef}
\end{equation}
where $\psi_c(\cdot)$ is the potential function, $\mathbf{x}_c$ is the set of variables defined on clique $c$, and $\mathcal{C}$ is the set of all the cliques. The most recent computer vision problems often require the energy models to contain a massive number of variables, have large label spaces and/or include higher-order interactions, which invariably enlarge the scale of the graph and increase the number of calls of graph cuts algorithm for solving the energy functions in the form (\ref{eq:energyDef}). Consequently, determining how to increase the scalability and speed up the optimization process of graph cuts becomes an urgent task. However, processor makers favor multi-core chip designs to manage CPU power dissipation. Therefore, fully exploiting the modern multi-core/multi-processor computer architecture to further boost the efficiency and scalability of graph cuts has attracted much attention recently.

\subsection{Related work}
Thanks to the duality relationship between maxflow and minimal $s-t$ cut, the maxflow algorithms, usually categorized as augmenting-path-based~\cite{sleator1981data,boykov2004experimental,ford1956maximal,goldberg2011maximum} and push-relabel-based~{\cite{goldberg1988new,cherkassky1995implementing,hochbaum2008pseudoflow,goldberg2008partial}, have been a focus for parallelizing in the literature.

Push-relabel-based methods~{\cite{goldberg1988new,cherkassky1995implementing,hochbaum2008pseudoflow,goldberg2008partial} are relatively easy to parallelize due to their memory locality feature. Therefore, a number of parallelized maxflow algorithms have been developed~{\cite{goldberg1991processor,bader2005cache,anderson1995parallel,delong2008scalable}. These algorithms generally exhibit superior performance on huge 3D grids with high connectivity. Nevertheless, they are usually unable to achieve the same efficiency on 2D grids or moderately sized sparse 3D grids, which is common in many computer vision problems, as the state-of-the-art serial augmenting path method~\cite{boykov2004experimental,goldberg2011maximum} could do on a commodity multi-core platform. The GPU implementation~\cite{vineet2008cuda} often fails to produce the correct results unless the amount of regularization is low~\cite{strandmark2010parallel}. The data transfer between the main memory and the graphics card is also a bottleneck for the GPU-based implementation.

Perhaps the most widely used graph cuts solver in the computer vision community is the algorithm proposed by Boykov and Kolmogorov~\cite{boykov2004experimental} (called the BK-algorithm). This is a serial augmenting path maxflow algorithm that effectively reuses the two search trees originated from $s$ and $t$, respectively. To parallelize the BK-algorithm for further efficiency, the graph is usually split into multiple parts, either disjoint or overlapping. Liu and Sun~\cite{liu2010parallel} uniformly partitioned the graph into a number of disjoint subgraphs, concurrently ran the BK-algorithm in each subgraph to obtain short-range search trees within each subgraph, and then adaptively merged adjacent subgraphs to enlarge the search range until only one subgraph remained and all augmenting paths were found. For some 2D image segmentation cases, this algorithm can achieve a near-linear speedup with up to $4$ computational threads. However, this method requires a shared-memory model, which makes it difficult to use on distributed platforms. To make the algorithm applicable to both shared and distributed models, Strandmark and Kahl~\cite{strandmark2010parallel,strandmark2011parallel} proposed a new parallel and distributed BK-algorithm (called the parallel BK-algorithm), which splits the graph into overlapped subgraphs based on dual decomposition~\cite{everett1963generalized,komodakis2007mrf}. The BK-algorithm is then run in a parallel and iterative fashion. Unfortunately, due to the possible multiple optimal solutions of the BK-algorithm, which is an inherent characteristic of all graph cuts methods, the parallel BK-algorithm may fail to converge in some cases~\cite{shekhovtsov-12-mixed_maxflow}.

By combining push-relabel and path augmentation methods, Shekhovtsov and Hlavac~\cite{shekhovtsov-12-mixed_maxflow} proposed an algorithm to further reduce the number of message exchanges in the region push-relabel algorithm~\cite{delong2008scalable}. Bhusnurmath and Taylor~\cite{bhusnurmath2008graph} reformulated graph cuts as an $\ell_1$ minimization problem and provided a highly parallelized implementation. However, even its GPU-based implementation is not significantly faster than the BK-algorithm for any type of graph.

\section{Convergence problem in the parallel BK-algorithm}
Because our work intends to remedy the non-convergence problem of the parallel BK-algorithm~\cite{strandmark2010parallel,strandmark2011parallel}, a brief review of it is first provided, followed by a convergence analysis of the parallel BK-algorithm. 

\subsection{A brief review of the parallel BK-algorithm}
By formulating the graph cuts problem defined on graph $G(V,C)$, with vertex set $V=\{s,t\}\cup\mathcal{V}$ and the capacity set $C$, as a linear program $E_{V}(\mathbf{x})$, Strandmark and Kahl~\cite{strandmark2010parallel,strandmark2011parallel} introduced a splitting method to split graph $G(V,C)$ into a set of overlapping subgraphs. Without loss of generality, only the two-subgraph case is discussed here, denoted as $G_1(V_1,C_1)$ and $G_2(V_2,C_2)$, where $V_1=\{s,t\}\cup\mathcal{V}_1$, $V_2=\{s,t\}\cup\mathcal{V}_2$ and  $\mathcal{V}_1\cap\mathcal{V}_2\ne\varnothing$. If the separable condition holds, stated as $\forall c_{ij}>0, \exists V_k\implies i\in V_k\land j\in V_k$, finding the $s-t$ cut with minimal cost in graph $G(V,C)$ can be reformulated as maximizing a concave non-differential dual function $g(\bm{\lambda})$, and it can be solved in the following iterative way: 1) calculate the optimal solutions $\accentset{*}{\mathbf{x}}_{\mathcal{V}_1}$ and $\accentset{*}{\mathbf{x}}_{\mathcal{V}_2}$ from the two sub-problems $E_{V_1}(\mathbf{x^1}\vert \bm{\lambda})$ and $E_{V_2}(\mathbf{x^2}\vert -\bm{\lambda})$ defined on $G_1(V_1,C_1)$ and $G_2(V_2,C_2)$, respectively, which are in the same form as the linear program representation of graph cuts such that the BK-algorithm can be applied on these two subgraphs simultaneously to obtain an ascent direction; 2) update the dual variables $\bm\lambda$ along this ascent direction to modify $E_{V_1}(\mathbf{x^1}\vert \bm{\lambda})$ and $E_{V_2}(\mathbf{x^2}\vert -\bm{\lambda})$ to initiate the next iteration. The part of Fig. \ref{fig:twoSubgraphs} before ``Merge Subgraphs'' illustrates this process, where the linear program representation $E$ for each graph (subgraph) is depicted above its corresponding box.

The above process is repeated until the following stopping rule is satisfied:
\begin{equation}
\label{eq:stopCriterion}
\left[\accentset{*}{\mathbf{x}}_{\mathcal{V}_1}^{(k)}- \accentset{*}{\mathbf{x}}_{\mathcal{V}_2}^{(k)}\right]_{\mathcal{V}_1\cap\mathcal{V}_2} = \bm 0,
\end{equation}
where $\left[\,\bm\cdot\,\right]_{\mathcal{V}_1\cap\mathcal{V}_2}$ is the projection on $\mathcal{V}_1\cap\mathcal{V}_2$. Once the above condition holds, the supergradient of $g(\bm{\lambda})$ equals 0, which means the maximum value is reached. 

\begin{figure}[htbp]
	\begin{center}
		\subfloat[One thread, 0.060 s]{\label{fig:oneThread}\includegraphics[width=0.49\linewidth]{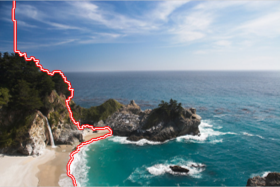}}\,
		\subfloat[Two threads, 0.047 s]{\label{fig:twoThreads}\includegraphics[width=0.49\linewidth]{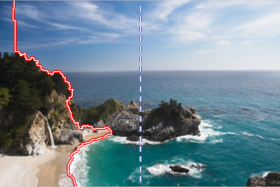}}\\
		\subfloat[Three threads, 0.138 s]{\label{fig:threeThreads}\includegraphics[width=0.49\linewidth]{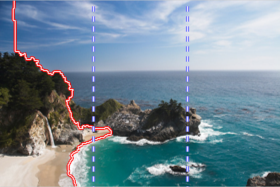}}\,
		\subfloat[Four threads, $\infty$]{\label{fig:fourThreads}\includegraphics[width=0.49\linewidth]{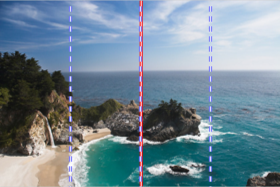}}
	\end{center}
	\caption{Solving an identical graph cuts problem using the parallel BK-algorithm with different numbers of threads.}
	\label{fig:diffNumThreads}
\end{figure}

\subsection{Convergence problem}\label{sec:convergence}

In theory, the parallel BK-algorithm is intended to maximize a concave non-differential function $g(\bm{\lambda})$ whose global optimal solution can be found by the supergradient ascent method. However, this algorithm may have problems converging to the optimal solution for the following two reasons: 1) Its sub-problems, $E_{V_1}(\mathbf{x^1}\vert \bm{\lambda})$ and $E_{V_2}(\mathbf{x^2}\vert -\bm{\lambda})$, may have multiple optimal solutions. Fig.~3 of Strandmark and Kahl's original paper~\cite{strandmark2010parallel} gives such an example. 2) The heuristic step size rule used by the parallel BK-algorithm, has a much faster convergence rate in practice, but it has no theoretical guarantee of convergence. 

Fig.~\ref{fig:diffNumThreads} shows an example of failure of the parallel BK-algorithm. Note that unless otherwise stated, all the experiments in this paper are conducted on a Sun Grid Engine cluster, with each compute node having $2\times$Intel\textsuperscript{\textregistered} Xeon\textsuperscript{\textregistered} E5-2670 v2 @ 2.5GHz (up to 20 Cores and 40 Threads) and 128GB memory. The graph associated with the image of Fig. \ref{fig:diffNumThreads} is constructed as follows: The edge costs are determined by the image gradient and the pixels in the leftmost and rightmost columns are connected to the source and the sink respectively. The red curves in Fig. \ref{fig:diffNumThreads} depict the final ``cut'', and the dotted blue lines are the boundaries of the neighboring parts. Fig. \subref*{fig:oneThread} is the result of the serial BK-algorithm. Its final ``cut'' serves as the ground-truth and its running time serves as the baseline in this example. Vertical splitting would lead to the ``worst-case scenario'', since all the possible $s-t$ paths are severed and all flows have to be communicated between the threads. Part of this example comes from Fig.~8 of Strandmark and Kahl's original paper~\cite{strandmark2010parallel}, where it is used to demonstrate the robustness of the parallel BK-algorithm to a poor splitting choice. They showed that if the graph is vertically split into only two parts, the parallel BK-algorithm could still obtain the same ``cut'', but could do so approximately $22\%$ faster than the serial BK-algorithm, as shown in Fig.~\subref*{fig:twoThreads}. However, we find that if the number of split parts is further increased, the performance of the parallel BK-algorithm begins to degrade. With three threads, as shown in Fig. \subref*{fig:threeThreads}, the parallel BK-algorithm takes more than twice the computational time of the serial BK-algorithm to obtain the same ``cut''. If the graph is vertically split into four parts, as shown in Fig. \subref*{fig:fourThreads}, the parallel BK-algorithm is unable to converge. 

To further assess the impact of the methods of splitting and the number of computational threads on the convergence of the parallel BK-algorithm, the following two fore-/back-ground segmentation problems are carried out on the enlarged Berkeley segmentation dataset~\cite{amfm_pami2011}, which consists of 500 images. The graph construction for the first segmentation problem, denoted as seg1, is the same as the above experiment. The graph construction for the second segmentation problem, denoted as seg2, is a lightly different from seg1. In seg2, each pixel is connected to both the source and the sink, and the edge capacity depends on the value of the pixel. Therefore, vertically splitting is a reasonable splitting choice. We tried 2 to 8 computational threads of the parallel BK-algorithm to solve these two segmentation problems on all 500 images. The number of images that failed to converge in each case are presented in Table \ref{tab:numFailedExamples}. The maximum allowed number of iterations is set to $1000$, which is the same as that used in its original implementation.

\begin{table}
    \caption{Number of images that failed to converge in the parallel BK-algorithm, where $\mathrm{PRB}$ represents the segmentation problems and $\mathrm{TRDS}$ represents the number of computational threads}
    \label{tab:numFailedExamples}
    \begin{center}
	\begin{tabular}{l|p{0.2cm}p{0.3cm}*{5}{p{0.45cm}}}
	    \toprule
	    \diagbox{\footnotesize PRB}{\footnotesize TRDS} & 2 & 3 & 4 & 5 & 6 & 7 & 8 \\
	    \midrule
	    seg1  & 0 & 16 & 500 & 500 & 500 & 500 & 500 \\
	    seg2  & 0 & 5 & 14 & 30 & 67 & 92 & 138 \\
	    \bottomrule
	\end{tabular}
    \end{center}
\end{table}

It can be seen from Table \ref{tab:numFailedExamples} that improper splitting and the thread number both have a great impact on the convergence of the parallel BK-algorithm, where all the images failed to converge in seg1 once the thread number exceeded 3. Even if the graph is split properly, the failure of convergence seems to be unavoidable with a moderate number of computation threads. In seg2, where the vertical splitting was reasonable, for approximately $3\%$ of the images, the method failed to converge with 4 computational threads, and the failure rate increased to nearly $30\%$ with 8 computational threads. The lack of convergence guarantee is a great deficiency of the parallel BK-algorithm, which severely hampers its applicability.

\section{A merging method} \label{sec:mergingMethod}
To remedy the non-convergence problem in the parallel BK-algorithm, we first introduce a merging method that can deal with the suspected neighboring subgraphs causing non-convergence and then give a correctness and efficiency analysis of the merging method based on pseudo-boolean representations of graph cuts. 


\begin{figure}[htbp]
    \begin{center}
	\captionsetup{justification=raggedright}
	\subfloat[$G_1(V_1,C_1)$]{\label{fig:mergeSub1}\includegraphics[height=0.4\linewidth]{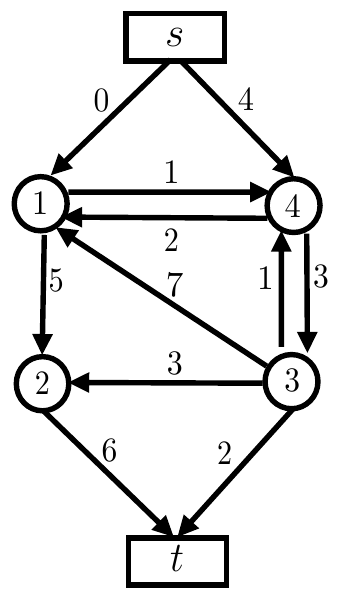}}\quad
	\subfloat[$G_2(V_2,C_2)$]{\label{fig:mergeSub2}\includegraphics[height=0.4\linewidth]{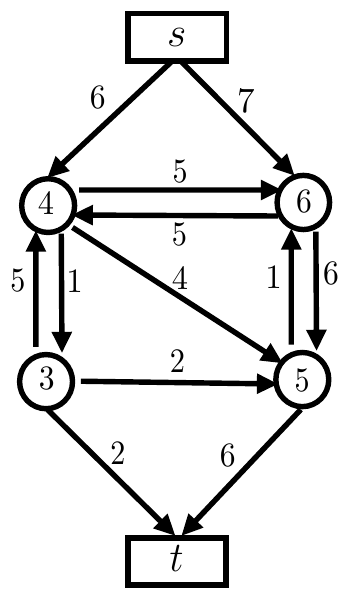}}\quad
	\subfloat[$G_{12}(V_{12},C_{12})$]{\label{fig:mergeGraph}\includegraphics[height=0.4\linewidth]{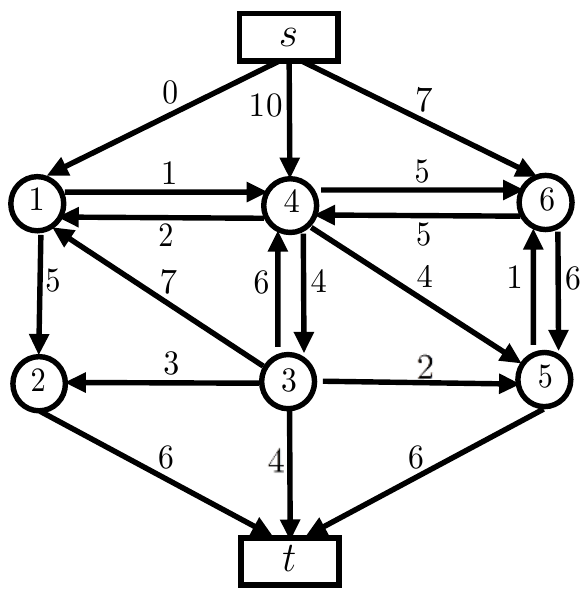}}
    \end{center}
    \caption{Merging two neighboring subgraphs $G_1(V_1,C_1)$ and $G_2(V_2,C_2)$ into a single graph $G_{12}(V_{12},C_{12})$.}
    \label{fig:mergeSubgraphs}
\end{figure}

As indicated by the experiments in the previous section, splitting a graph into more subgraphs often increases the speedup, but it also increases the possibility of non-convergence. However, it is difficult to know the optimal number of subgraphs at the beginning of the parallel BK-algorithm. Is it possible to dynamically adjust the current splitting of a graph when over-splitting is detected during the running time, so as to further speed up the computation while simultaneously guaranteeing convergence? To this end, a merging method is proposed in this paper.

Merging two subgraphs into a single graph is performed as follows: given two neighboring subgraphs $G_1(V_1,C_1)$ and $G_2(V_2,C_2)$, denote the merged graph as $G_{12}(V_{12},C_{12})$, where $V_1=\{s,t\}\cup\mathcal{V}_1$, $V_2=\{s,t\}\cup\mathcal{V}_2$, $\mathcal{V}_1\cap \mathcal{V}_2\ne\varnothing$ and $V_{12}=\{s,t\}\cup\mathcal{V}_{12}$. The merged vertex set $\mathcal{V}_{12}=\mathcal{V}_1\cup\mathcal{V}_2$. Only if the two endpoints are all within $V_1\cap V_2$ is the edge capacity of the merged graph $G_{12}$ the summation of the capacities of the same edge in $G_1$ and $G_2$. Otherwise, the edge capacity of the merged graph $G_{12}$ is simply that from only one of the two subgraphs. Fig.~\ref{fig:mergeSubgraphs} gives an example. Merging more than two subgraphs into a single graph can be performed similarly, where the vertices in the overlapped region are absorbed and the edge capacities in the overlapped region are summed. It is worth noting that the merging method can also be applied to N-D graphs (subgraphs), just like the splitting method of the parallel BK-algorithm.

Fig. \ref{fig:twoSubgraphs} shows the application of the merging method in the parallel BK-algorithm with only two subgraphs. When there is evidence that the two subgraphs may have difficulties reaching agreement on their optimal values, the merging operation is invoked to obtain a merged graph $G'$ from which the optimal values $\accentset{*}{\mathbf{x}}_{\mathcal{V}}$ are obtained. This naturally raises the following key question: will the merged graph $G'$ obtain the same optimal solutions as the original graph $G$, despite that the original was first split into two subgraphs, both of which then experienced maxflow computation and updating of many rounds? 

\begin{figure*}[htbp]
    \begin{center}
	\includegraphics[width=0.99\linewidth]{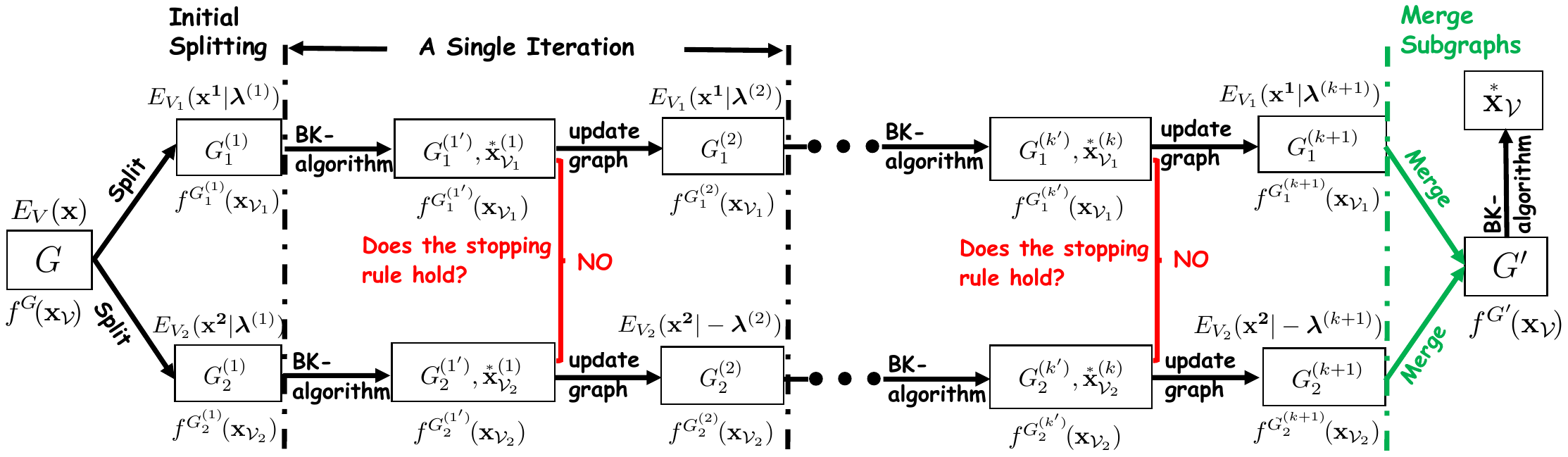}
    \end{center}
    \caption{Application of the merging method in the parallel BK-algorithm with only two subgraphs.}
    \label{fig:twoSubgraphs}
\end{figure*}
There are two related questions: How should one decide which subgraphs to merge? What is the best time for merging? Although there are many cues for over-splitting, such as the existence of a large number of vertices within the overlapped region that admit multiple optimal solutions or the scarcity of pushed flows in the subgraphs, there is currently no clear-cut answer. However, we use a simple strategy based on the following assumption: the number of vertices in the overlapped region that disagree on their optimal values is expected to decrease with iterations if the current graph splitting is proper. The strategy is detailed in Algorithm \ref{alg:naiveConvergedParallel}, where $numDiff$ is used to record the minimal number of disagreement vertices of all the iterations. If $numDiff$ remains non-decreasing in some successive iterations, which is controlled by $\mathrm{ITER}$, then merging should be invoked to ensure better graph splitting. In the merging operation, every two adjacent subgraphs are merged into one graph in order to maintain work load balance. Clearly, the number of disagreement vertices will decrease at least by 1 for every $\mathrm{ITER}-1$ iterations, or every two adjacent subgraphs are merged. Therefore, Algorithm \ref{alg:naiveConvergedParallel} will converge within no more than $M\times(\mathrm{ITER}-1) + \log_2^{N}$ iterations, where $M$ is the number of disagreement vertices in the first iteration and $N$ is the number of subgraphs into which the original graph is split in the initial splitting stage. Since a naive merging strategy is used and the convergence guarantee holds, Algorithm \ref{alg:naiveConvergedParallel} is named the ``Naive Converged Parallel BK-Algorithm''. Obviously, more-suitable merging strategies can further boost the performance, this issue will be investigated in our further work.

\begin{algorithm}[htbp]                      
	\caption{Naive Converged Parallel BK-Algorithm}          
	\label{alg:naiveConvergedParallel}                           
	\begin{algorithmic}[1]                    
		\STATE \textbf{{\bf Set} $numDiff:=+\infty$;  $iter:=0$;}
		\STATE {Split graph $G$ into $N$ overlapped subgraphs $G_1,\dots,G_N$;}
		\REPEAT
		\STATE{Run the BK-algorithm concurrently on all the subgraphs $G_i^{(k)}$ to obtain $\accentset{*}{\mathbf{x}}_{\mathcal{V}_i}^{(k)}$£» } 
		\STATE{Update all the residual subgraphs $G_i^{(k')}$ to obtain $G_i^{(k+1)}$;}
		\STATE{Count the number of nodes that disagree on their optimal values, denoted as $nDiff$;}
		\IF{$nDiff \ge numDiff$}
		\STATE{$++iter$;}
		\IF{$iter == \mathrm{ITER}$}
		\STATE{Every two neighboring subgraphs are merged into one graph;}
		\ENDIF
		\ELSE
		\STATE{$numDiff:=nDiff$;}
		\STATE{$iter:=0$;}
		\ENDIF		
		\UNTIL{$nDiff == 0$;}
		\end{algorithmic}
\end{algorithm}

\section{Pseudo-boolean representation-based invariance analysis for graph cuts algorithms}
As noted in the previous section, the correctness of our merging methods depends on whether the merged graph obtains the same optimal solutions as that of the original graph. To answer this question, we propose a new pseudo-boolean representation, named \emph{restricted homogeneous posiforms}, to track the changes for all the graphs (subgraphs) under the operations of the parallel graph cuts algorithms. We then develop an invariance analysis method.

\subsection{Restricted homogeneous posiforms for graph cuts}

Since each $s-t$ cut in a graph is a partition of all its vertices to either $S$ or $T$, it can be represented by a realization of a set of boolean variables, $\mathbf{x}_{\mathcal{V}}=\{x_i\vert i\in\mathcal{V}\}$. Therefore, a function with its arguments being boolean to represent an $s-t$ cut and its value being real representing the cost of this $s-t$ cut can be used to express a graph cuts problem. Such a function is usually called a pseudo-boolean function in the combinatorial optimization community~\cite{boros2002pseudo,ivuanescu1965some,Boros06preprocessingof}, whose representations are often categorized into two types: \emph{multi-linear polynomials} whose form is uniquely determined by a pseudo-boolean function and \emph{posiforms} that can uniquely determine a pseudo-boolean function, although a pseudo-boolean function can have many different posiforms representing it. Among these many possible posiforms, we define the \emph{restricted homogeneous posiforms}, which can have a one-to-one correspondence with the graph on which the graph cuts problem is defined. 
\begin{thm}
	\label{thm:homogeneousRep}
	A graph cuts problem defined on graph $G(V,C)$, with $V=\{s,t\}\cup\mathcal{V}$ and non-negative edge capacities, can be uniquely represented using the following restricted homogeneous posiforms:
	\begin{equation}
	\label{eq:homogeneousRep}
	\phi^G_h(\mathbf{x}_\mathcal{V}) = \sum_{i\in\mathcal{V}}a_ix_i + \sum_{j\in\mathcal{V}}a_j\bar x_j + \sum_{i,j\in\mathcal{V}}a_{ij}\bar x_ix_j,
	\end{equation}
	where $x_i\in\{0,1\},\bar x_i=1-x_i, a_i\ge0,a_j\ge0,a_{ij}\ge0$. Moreover, the components of $\phi^G_h(\mathbf{x}_\mathcal{V})$ have the following one-to-one correspondence with edges of graph $G(V,C)$, as shown in Table \ref{tab:correspondence}.
	\begin{table}[H]
		\centering
		\caption{The correspondence of components of $\phi^G_h$ and edges of $G$, where $x_i=0$ implies vertex $i$ belongs to $S$ and $x_i=1$ implies vertex $i$ belongs to $T$ in an $s-t$ cut.}
		\label{tab:correspondence}
		\begin{tabular}{c|c|c}
			\toprule
			Comp. $\phi^G_h(\mathbf{x}_\mathcal{V})$ & Edge $G(V,C)$ & Relations\\
			\midrule
			$a_ix_i$ & $(s,i)$ & $a_i=c_{si},\,\, c_{si}\in C$ \\
			$a_j\bar x_j$ & $(j,t)$ & $a_j=c_{jt},\,\, c_{jt}\in C$ \\
			$a_{ij}\bar x_ix_j$ & $(i,j)$ & $a_{ij}=c_{ij},\,\, c_{ij}\in C$ \\
			\bottomrule
		\end{tabular}
	\end{table}
\end{thm}

\begin{proof}	
	The cost of an $s-t$ cut is 
	\begin{equation}
	\label{eq:costSTCut}
	C_{S,T}=\sum_{i\in S, j\in T}c_{ij}
	\end{equation}
	and the relations between the values of $x_i$ and the parts to which the node $i$ belongs are:
	\begin{equation}
	x_i = 
	\begin{cases}
	0 &\!\!\!\iff i \in S \\
	1 &\!\!\!\iff i \in T
	\end{cases}.
	\label{eq:relationSTCut}
	\end{equation}
    Only if $\phi^G_h(\mathbf{x}_\mathcal{V})$ equals the cost, up to a constant, of the $s-t$ cut defined by the realization of $\mathbf{x}_\mathcal{V}$, for every possible realizations of $\mathbf{x}_\mathcal{V}$, can  
    $\phi^G_h(\mathbf{x}_\mathcal{V})$ be the pseudo-boolean function of the graph cuts problem defined on graph $G(V,C)$. Since Table \ref{tab:correspondence} defines a one-to-one relationship between the components of the restricted homogeneous posiforms and the edges of the graph $G(V,C)$, it suffices to verify the equivalence between the value of each component of $\phi^G_h(\mathbf{x}_\mathcal{V})$ and the cost of its corresponding edge in the $s-t$ cut under all the possible realizations of $\mathbf{x}_\mathcal{V}$. This can be done for the three components listed in Table \ref{tab:correspondence}. 

	1). $a_ix_i$
	
	The value of $a_ix_i$ is  
	$
	\begin{cases}
	0 &\!\!\!\iff x_i =0 \\
	a_i &\!\!\!\iff x_i = 1
	\end{cases}
	$, 
	its corresponding edge $(s,i)$ contributes to the cost of an $s-t$ cut by: 
	$
	\begin{cases}
	0 &\!\!\!\iff i \in S \\
	c_{si} &\!\!\!\iff i \in T
	\end{cases}
	$,
	which can be seen from the definition of the cost of an $s-t$ cut in (\ref{eq:costSTCut}). From the relations of (\ref{eq:relationSTCut}) and the equivalence of $a_i$ and $c_{si}$, it is clear that the value of $a_ix_i$ equals the cost of its corresponding edge $(s,i)$ in the $s-t$ cut under all the possible realizations of $\mathbf{x}_\mathcal{V}$.
	
	2). $a_j\bar x_j$
	
	It can be verified in a similar way as in the previous case that the value of $a_j\bar x_j$ equals the cost of its corresponding edge $(j,t)$ in the $s-t$ cut under all the possible realizations of $\mathbf{x}_\mathcal{V}$.
	
	3). $a_{ij}\bar x_ix_j$
	
	The value of $a_{ij}\bar x_ix_j$ is
	$
	\begin{cases}
	a_{ij} &\!\!\!\iff x_i =0, x_j = 1\\
	0 &\!\!\!\iff otherwise
	\end{cases}
	$.
	Its corresponding edge $(i,j)$ contributes to the cost of an $s-t$ cut by:
	$
	\begin{cases}
	c_{ij} &\!\!\!\iff i \in S, j \in T \\
	0 &\!\!\!\iff otherwise
	\end{cases}
	$.
	From the relations of (\ref{eq:relationSTCut}) and the equivalence of $a_{ij}$ and $c_{ij}$, it is clear that the value of $a_{ij}\bar x_ix_j$ equals the cost of its corresponding edge $(i,j)$ in the $s-t$ cut under all the possible realizations of $\mathbf{x}_\mathcal{V}$.

	For each of the three components in the restricted homogeneous posiforms (\ref{eq:homogeneousRep}), its value under all the possible realizations of $\mathbf{x}_\mathcal{V}$ equals the contribution of its corresponding edge to the cost of the $s-t$ cut. Therefore, Theorem \ref{thm:homogeneousRep} is proved.
\end{proof}

Theorem \ref{thm:homogeneousRep} states that each unary term in the restricted homogeneous posiforms, as in (\ref{eq:homogeneousRep}), corresponds to an t-link edge in the graph, whereas each pairwise term corresponds to an n-link edge, and their coefficients are the capacities of those edges in the graph. Therefore, any changes in the graph will be reflected in its restricted homogeneous posiforms, and vice versa.

In addition to the above proposed restricted homogeneous posiforms, a pseudo-boolean function can also be uniquely represented as a multi-linear polynomial~\cite{boros2002pseudo,ivuanescu1965some,Boros06preprocessingof}:
\begin{equation}
\label{eq:polynomial}
f^G(\mathbf{x}) = l_0 + \sum_{j=1}^nl_ix_i + \sum_{ij} l_{ij}x_ix_j,\,\,\, x_i\in\{0,1\},
\end{equation}
with $l_{ij} \le 0$. In fact, $l_{ij} \le 0$ is the necessary and sufficient condition for $f^G(\mathbf{x})$ to represent a graph cuts problem defined on a graph $G$ with non-negative edge capacities, as proved by Freeman and Drineas~\cite{Freedman05}

\begin{figure}[!htbp]
	\begin{center}
		\subfloat[Partition of the set $\mathcal{V}$]{\label{fig:setPartition}\includegraphics[height=0.48\linewidth]{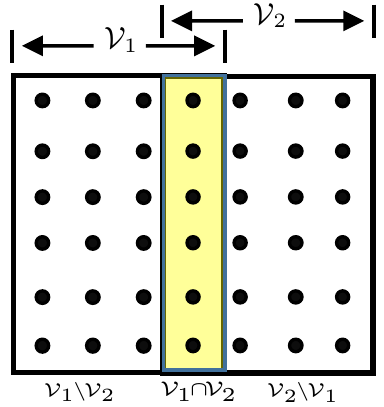}}\quad\,
		\subfloat[$G(V,C)$]{\label{fig:mergedGraph}\includegraphics[height=0.48\linewidth]{mergedGraph_aurora.pdf}}\\
		\subfloat[$G_1(V_1,C_1)$]{\label{fig:splitGraphL}\includegraphics[height=0.48\linewidth]{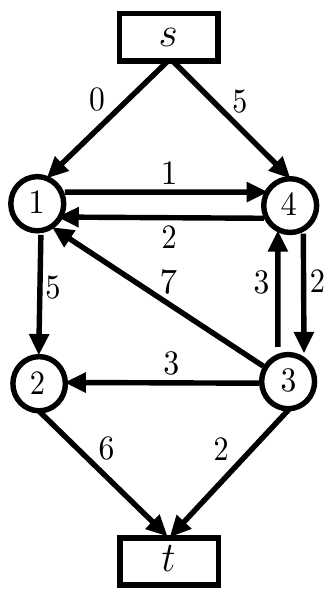}}\qquad\qquad\qquad
		\subfloat[$G_2(V_2,C_2)$]{\label{fig:splitGraphR}\includegraphics[height=0.48\linewidth]{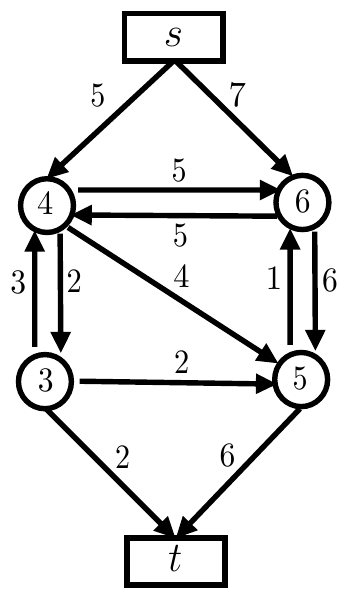}}
	\end{center}
	\caption{Splitting graph $G(V,C)$ into two subgraphs $G_1(V_1,C_1)$ and $G_2(V_2,C_2)$ .}
	\label{fig:graphSplitting}
\end{figure}
\subsection{Invariance analysis for parallel graph cuts algorithm}
By the above multi-linear polynomials and proposed restricted homogeneous posiforms, here we give an invariance analysis method for parallel graph cuts algorithm.

There are four types of operations in the naive converged parallel BK-algorithm: splitting, maxflow computation, graph updating and merging subgraphs. The first three types of operations are also used in the parallel BK-algorithm. Fig. \ref{fig:twoSubgraphs} is the pipeline of the naive converged parallel BK-algorithm in the two subgraphs case. Without loss of generality, only the two subgraphs case is analyzed here, but the conclusions can be easily generalized to more subgraphs cases. 

Suppose the original graph is $G$, the two subgraphs in the $k$th iteration are $G_1^{(k)}, G_2^{(k)}$, and their corresponding restricted homogeneous posiforms are $\phi^G_h(\mathbf{x}_\mathcal{V})$, $\phi^{G_1^{(k)}}_h(\mathbf{x}_{\mathcal{V}_1})$ and $\phi_h^{G_2^{(k)}}(\mathbf{x}_{\mathcal{V}_2})$ respectively. Since $V=\{s,t\}\cup\mathcal{V}$ and $\mathcal{V}_1\setminus\mathcal{V}_2$, $\mathcal{V}_1\cap\mathcal{V}_2$ and $\mathcal{V}_2\setminus\mathcal{V}_1$ are partitions of $\mathcal{V}$, as shown in Fig. \subref*{fig:setPartition}. The three restricted homogeneous posiforms can be equally written as:

\begin{equation}
\begin{split}
&{}\phi^G_h(\mathbf{x}_{\mathcal{V}})   =
\sum_{i\in \mathcal{V}_1\!\setminus\mathcal{V}_2}\!\!\!\!\widehat{\underaccent{\dot}{a}}_ix_i + \!\!\sum_{i\in \mathcal{V}_1\cap\mathcal{V}_2}\!\!\!\!\bar{\underaccent{\dot}{a}}_i x_i +\!\! \sum_{i\in \mathcal{V}_2\!\setminus\mathcal{V}_1}\!\!\!\!\widetriangle{\underaccent{\dot}{a}}_ix_i +  \\
&{}\sum_{i\in \mathcal{V}_1\!\setminus\mathcal{V}_2}\!\!\!\!\widehat{\underaccent{\ddot}{a}}_i\bar x_i + \!\!\sum_{i\in \mathcal{V}_1\cap\mathcal{V}_2}\!\!\!\!\bar{\underaccent{\ddot}{a}}_i\bar x_i +\!\! \sum_{i\in \mathcal{V}_2\!\setminus\mathcal{V}_1}\!\!\!\!\widetriangle{\underaccent{\ddot}{a}}_i\bar x_i + \\
&{}\!\!\sum_{i,j\in\mathcal{V}_1\cap\mathcal{V}_2}\!\!\!\!\bar a_{ij}\bar x_ix_j +\!\! \sum_{i\!\text{ \it or }\!j\in \mathcal{V}_1\!\setminus\mathcal{V}_2}\!\!\!\!\widehat a_{ij}\bar x_ix_j +\!\! \sum_{i\!\text{ \it or }\!j\in \mathcal{V}_2\!\setminus\mathcal{V}_1}\!\!\!\!\widetriangle a_{ij}\bar x_ix_j\,,
\end{split}
\end{equation}

\begin{equation}
\begin{split}
\phi^{G_1^{(k)}}_h(\mathbf{x}_{\mathcal{V}_1})  &= \sum_{i\in \mathcal{V}_1\!\setminus\mathcal{V}_2}\!\!\!\!\widehat{\underaccent{\dot}{a}}^{1^{(k)}}_ix_i + \!\!\sum_{i\in \mathcal{V}_1\cap\mathcal{V}_2}\!\!\!\!\bar{\underaccent{\dot}{a}}^{1^{(k)}}_ix_i + \\
&\mathrel{\phantom{=}} \sum_{i\in \mathcal{V}_1\!\setminus\mathcal{V}_2}\!\!\!\!\widehat{\underaccent{\ddot}{a}}^{1^{(k)}}_i\bar x_i + \!\!\sum_{i\in \mathcal{V}_1\cap\mathcal{V}_2}\!\!\!\!\bar{\underaccent{\ddot}{a}}^{1^{(k)}}_i\bar x_i + \\
&\mathrel{\phantom{=}}\!\!\! \sum_{i,j\in\mathcal{V}_1\cap\mathcal{V}_2}\!\!\!\!\bar a^{1^{(k)}}_{ij}\bar x_ix_j +\!\! \sum_{i\!\text{ \it or }\!j\in \mathcal{V}_1\!\setminus\mathcal{V}_2}\!\!\!\!\widehat a^{1^{(k)}}_{ij}\bar x_ix_j\,,
\end{split}
\end{equation}

\begin{equation}
\begin{split}
\phi^{G_2^{(k)}}_h(\mathbf{x}_{\mathcal{V}_2})  &= \sum_{i\in \mathcal{V}_2\!\setminus\mathcal{V}_1}\!\!\!\!\widetriangle{\underaccent{\dot}{a}}^{2^{(k)}}_ix_i + \!\!\sum_{i\in \mathcal{V}_1\cap\mathcal{V}_2}\!\!\!\!\bar{\underaccent{\dot}{a}}^{2^{(k)}}_ix_i + \\
&\mathrel{\phantom{=}} \sum_{i\in \mathcal{V}_2\!\setminus\mathcal{V}_1}\!\!\!\!\widetriangle{\underaccent{\ddot}{a}}^{2^{(k)}}_i\bar x_i + \!\!\sum_{i\in \mathcal{V}_1\cap\mathcal{V}_2}\!\!\!\!\bar{\underaccent{\ddot}{a}}^{2^{(k)}}_i\bar x_i + \\
&\mathrel{\phantom{=}}\!\!\! \sum_{i,j\in\mathcal{V}_1\cap\mathcal{V}_2}\!\!\!\!\bar a^{2^{(k)}}_{ij}\bar x_ix_j +\!\! \sum_{i\!\text{ \it or }\!j\in \mathcal{V}_2\!\setminus\mathcal{V}_1}\!\!\!\!\widetriangle a^{2^{(k)}}_{ij}\bar x_ix_j\,\,. 
\end{split}
\end{equation}

The multi-linear polynomials of graphs $G$, $G_1^{(k)}$ and $G_2^{(k)}$, denoted as $f^{G}(\mathbf{x}_{\mathcal{V}})$, $f^{G_1^{(k)}}(\mathbf{x}_{\mathcal{V}_1})$ and $f^{G_2^{(k)}}(\mathbf{x}_{\mathcal{V}_2})$ respectively, can be expressed as:
\begin{equation}
\begin{split}
&{}f^{G}(\mathbf{x}_{\mathcal{V}}) = l + \!\! \sum_{i\in\mathcal{V}_1\!\setminus \mathcal{V}_2}\!\!\widehat{l}_ix_i + \!\! \sum_{i\in\mathcal{V}_1\cap \mathcal{V}_2}\!\!\bar l_ix_i +  \!\! \sum_{i\in\mathcal{V}_2\!\setminus\mathcal{V}_1}\!\!\widetriangle l_ix_i +\\
&{}\sum_{i,j\in\mathcal{V}_1\cap\mathcal{V}_2}\!\!\!\!\bar l_{ij}x_ix_j +\!\! \sum_{i\!\text{ \it or }\!j\in\mathcal{V}_1\!\setminus \mathcal{V}_2}\!\!\!\!\widehat l_{ij}x_ix_j +\!\! \sum_{i\!\text{ \it or }\!j\in\mathcal{V}_2\!\setminus \mathcal{V}_1}\!\!\!\!\widetriangle l_{ij}x_ix_j \,\,,
\end{split}
\end{equation}

\begin{equation}
\begin{split}
f^{G_1^{(k)}}(\mathbf{x}_{\mathcal{V}_1}) &= l^{1^{(k)}}\!\! + \!\! \sum_{i\in\mathcal{V}_1\!\setminus\mathcal{V}_2}\!\!\widehat l^{1^{(k)}}_ix_i + \!\! \sum_{i\in\mathcal{V}_1\cap \mathcal{V}_2}\!\!\bar l^{1^{(k)}}_ix_i +\\
&\sum_{i,j\in\mathcal{V}_1\cap \mathcal{V}_2}\!\!\!\!\bar l^{1^{(k)}}_{ij}x_ix_j +\!\! \sum_{i\!\text{ \it or }\!j\in\mathcal{V}_1\!\setminus\mathcal{V}_2}\!\!\!\!\widehat l^{1^{(k)}}_{ij}x_ix_j\,\,, 
\end{split}
\end{equation}

\begin{equation}
\begin{split}
f^{G_2^{(k)}}(\mathbf{x}_{\mathcal{V}_2}) &= l^{2^{(k)}}\!\! + \!\! \sum_{i\in\mathcal{V}_2\!\setminus\mathcal{V}_1}\!\!\!\!\widetriangle l^{2^{(k)}}_ix_i + \!\! \sum_{i\in\mathcal{V}_1\cap\mathcal{V}_2}\!\!\!\!\bar l^{2^{(k)}}_ix_i +\\
&\sum_{i,j\in\mathcal{V}_1\cap\mathcal{V}_2}\!\!\!\!\!\!\bar l^{2^{(k)}}_{ij}x_ix_j +\!\! \sum_{i\!\text{ \it or }\!j\in\mathcal{V}_2\!\setminus\mathcal{V}_1}\!\!\!\!\!\!\widetriangle l^{2^{(k)}}_{ij}x_ix_j\,\,. 
\end{split}
\end{equation}

It is easy to verify that the following relations hold for the coefficients of $f^{G}(\mathbf{x}_{\mathcal{V}})$ and of $\phi^G_h(\mathbf{x}_\mathcal{V})$:
\begin{equation}
\label{eq:relationFPhi^}
l = \sum_{i\in\mathcal{V}_1\!\setminus\!\mathcal{V}_2}\!\!\widehat{\underaccent{\ddot}{a}}_i + \sum_{i\in\mathcal{V}_1\cap\mathcal{V}_2}\!\!\bar{\underaccent{\ddot}{a}}_i + \sum_{i\in\mathcal{V}_2\!\setminus\!\mathcal{V}_1}\!\!\widetriangle{\underaccent{\ddot}{a}}_i\,\,,
\end{equation}

\begin{equation}
\widehat l_i = \widehat {\underaccent{\dot}{a}}_i - \widehat{\underaccent{\ddot}{a}}_i + \sum_{k\in\mathcal{V}_1}\widehat a_{ki}, \forall i\in\mathcal{V}_1\!\setminus\!\mathcal{V}_2\,\,,
\end{equation}

\begin{equation}
\bar l_i = \bar {\underaccent{\dot}{a}}_i -\bar {\underaccent{\ddot}{a}}_i + \sum_{k\in\mathcal{V}_1\cap\mathcal{V}_2}\!\!\!\!\bar a_{ki}, \forall i\in\mathcal{V}_1\!\cap\!\mathcal{V}_2\,\,,
\end{equation}

\begin{equation}
\widetriangle l_i = \widetriangle {\underaccent{\dot}{a}}_i - \widetriangle{\underaccent{\ddot}{a}}_i + \sum_{k\in\mathcal{V}_2}\widetriangle a_{ki},\forall i\in \mathcal{V}_2\!\setminus\!\mathcal{V}_1\,\,,
\end{equation}

\begin{equation}
\bar l_{ij} = -\bar a_{ij},\forall i,j\in\mathcal{V}_1\!\cap\!\mathcal{V}_2\,\,,
\end{equation}

\begin{equation}
\widehat l_{ij} = -\widehat a_{ij},\forall i\!\text{ \it or }\!j\in\mathcal{V}_1\!\setminus\!\mathcal{V}_2\,\,,
\end{equation}

\begin{equation}
\label{eq:relationFPhi$}
\widetriangle l_{ij} = -\widetriangle a_{ij},\forall i\!\text{ \it or }\!j\in \mathcal{V}_2\!\setminus\!\mathcal{V}_1\,\,.
\end{equation}

Similarly, the following relations hold for the coefficients of $f^{G_1^{(k)}}(\mathbf{x}_{\mathcal{V}_1})$ and of $\phi^{G_1^{(k)}}_h(\mathbf{x}_{\mathcal{V}_1})$:
\begin{equation}
\label{eq:relationF1Phi1^}
l^{1^{(k)}} = \sum_{i\in\mathcal{V}_1\!\setminus\!\mathcal{V}_2}\!\!\widehat{\underaccent{\ddot}{a}}_i^{1^{(k)}} + \sum_{i\in\mathcal{V}_1\cap\mathcal{V}_2}\!\!\bar{\underaccent{\ddot}{a}}_i^{1^{(k)}}\,\,,
\end{equation}

\begin{equation}
\widehat l_i^{1^{(k)}} = \widehat {\underaccent{\dot}{a}}_i^{1^{(k)}}\!\! - \widehat{\underaccent{\ddot}{a}}_i^{1^{(k)}} + \sum_{k\in\mathcal{V}_1}\widehat a_{ki}^{1^{(k)}}, \forall i\in\mathcal{V}_1\!\setminus\!\mathcal{V}_2\,\,,
\end{equation}

\begin{equation}
\label{eq:relationF1Phi1barNode}
\bar l_i^{1^{(k)}} = \bar {\underaccent{\dot}{a}}_i^{1^{(k)}}\!\! -\bar {\underaccent{\ddot}{a}}_i^{1^{(k)}} + \sum_{k\in\mathcal{V}_1\cap\mathcal{V}_2}\!\!\!\!\bar a_{ki}^{1^{(k)}}, \forall i\in\mathcal{V}_1\!\cap\!\mathcal{V}_2\,\,,
\end{equation}

\begin{equation}
\bar l_{ij}^{1^{(k)}} = -\bar a_{ij}^{1^{(k)}},\forall i,j\in\mathcal{V}_1\!\cap\!\mathcal{V}_2\,\,,
\end{equation}

\begin{equation}
\label{eq:relationF1Phi1$}
\widehat l_{ij}^{1^{(k)}} = -\widehat a_{ij}^{1^{(k)}},\forall i\!\text{ \it or }\!j\in\mathcal{V}_1\!\setminus\!\mathcal{V}_2\,\,.
\end{equation}

And those between $f^{G_2^{(k)}}(\mathbf{x}_{\mathcal{V}_2})$ and $\phi_h^{G_2^{(k)}}(\mathbf{x}_{\mathcal{V}_2})$:
\begin{equation}
\label{eq:relationF2Phi2^}
l^{2^{(k)}} = \sum_{i\in\mathcal{V}_2\!\setminus\!\mathcal{V}_1}\!\!\widehat{\underaccent{\ddot}{a}}_i^{2^{(k)}} + \sum_{i\in\mathcal{V}_1\cap\mathcal{V}_2}\!\!\bar{\underaccent{\ddot}{a}}_i^{2^{(k)}}\,\,,
\end{equation}

\begin{equation}
\widehat l_i^{2^{(k)}} = \widehat {\underaccent{\dot}{a}}_i^{2^{(k)}}\!\! - \widehat{\underaccent{\ddot}{a}}_i^{2^{(k)}} + \sum_{k\in\mathcal{V}_2}\widehat a_{ki}^{2^{(k)}}, \forall i\in\mathcal{V}_2\!\setminus\!\mathcal{V}_1\,\,,
\end{equation}

\begin{equation}
\label{eq:relationF2Phi2barNode}
\bar l_i^{2^{(k)}} = \bar {\underaccent{\dot}{a}}_i^{2^{(k)}}\!\! -\bar {\underaccent{\ddot}{a}}_i^{2^{(k)}} + \sum_{k\in\mathcal{V}_1\cap\mathcal{V}_2}\!\!\!\!\bar a_{ki}^{2^{(k)}}, \forall i\in\mathcal{V}_1\!\cap\!\mathcal{V}_2\,\,,
\end{equation}

\begin{equation}
\bar l_{ij}^{2^{(k)}} = -\bar a_{ij}^{2^{(k)}},\forall i,j\in\mathcal{V}_1\!\cap\!\mathcal{V}_2\,\,,
\end{equation}

\begin{equation}
\label{eq:relationF2Phi2$}
\widehat l_{ij}^{2^{(k)}} = -\widehat a_{ij}^{2^{(k)}},\forall i\!\text{ \it or }\!j\in\mathcal{V}_2\!\setminus\!\mathcal{V}_1\,\,.
\end{equation}

Thus far, the relations between the coefficients of the restricted homogeneous posiforms $\phi^G_h(\mathbf{x}_\mathcal{V})$, $\phi^{G_1^{(k)}}_h(\mathbf{x}_{\mathcal{V}_1})$, $\phi_h^{G_2^{(k)}}(\mathbf{x}_{\mathcal{V}_2})$ and those of the multi-linear polynomials $f^{G}(\mathbf{x}_{\mathcal{V}})$, $f^{G_1^{(k)}}(\mathbf{x}_{\mathcal{V}_1})$, $f^{G_2^{(k)}}(\mathbf{x}_{\mathcal{V}_2})$ have been established. The following four propositions state the changes of the multi-linear polynomials for each of the four operations in Fig. \ref{fig:twoSubgraphs}.

\begin{prop}
	\label{prop:split}
	The coefficients of $f^{G_1^{(1)}}(\mathbf{x}_{\mathcal{V}_1})$, $f^{G_2^{(1)}}(\mathbf{x}_{\mathcal{V}_2})$ and $f^G(\mathbf{x}_{\mathcal{V}})$, which are the multi-linear polynomials of the split subgraph $G_1^{(1)}$, $G_2^{(1)}$ and the original graph $G$, respectively, satisfy the following relations:
	\begin{myEnumerate}
		\item $ l^{1^{(1)}}\!\! + l^{2^{(1)}} = l.$ 
		\item $ \widehat l^{1^{(1)}}_i \!\!= \widehat l_i , \forall i\in\mathcal{V}_1\!\setminus\! \mathcal{V}_2;\quad$   $\widehat l_{ij}^{1^{(1)}}\!\! = \widehat l_{ij},\forall i\!\text{ \it or }\!j\in\mathcal{V}_1\!\setminus\!\mathcal{V}_2.$
		\item  $ \widetriangle l^{2^{(1)}}_i \!\!= \widetriangle l_i , \forall i\in\mathcal{V}_2\!\setminus\!\mathcal{V}_1;\quad$ $\widetriangle l_{ij}^{2^{(1)}} \!\!= \widetriangle l_{ij},\forall i\!\text{ \it or }\!j\in\mathcal{V}_2\!\setminus\!\mathcal{V}_1.$
		\item $\bar l^{1^{(1)}}_i\!\! =\bar l^{2^{(1)}}_i =\frac{1}{2} \bar l_i , \forall i\in\mathcal{V}_1\!\cap\!\mathcal{V}_2;\quad$ $\bar l_{ij}^{1^{(1)}}\!\! = \bar l_{ij}^{2^{(1)}}\!\! =\frac{1}{2} \bar l_{ij},\forall i,j\in\mathcal{V}_1\!\cap\!\mathcal{V}_2.$
	\end{myEnumerate}
\end{prop}
\begin{proof}
    See Appendix \ref{app:split}.
\end{proof}

\begin{prop}
\label{prop:maxflow}
For all augmenting path maxflow algorithms, pushing a flow $\mathrm{F}$ from $s$ to $t$ in graph G results in a residual graph $G'$, where the multi-linear polynomial of the residual graph $G'$ differs from that of the original graph $G$ by a constant term of $\mathrm{F}$.
\end{prop}
\begin{proof}
    See Appendix \ref{app:maxflow}.
\end{proof}

\begin{prop}
\label{prop:update}
The coefficients of $f^{G_1^{(k+1)}}$, $f^{G_2^{(k+1)}}$ and those of $f^{G_1^{(k')}}$, $f^{G_2^{(k')}}$, which are updated by the dual variables $\bm\lambda^{(k)}$, satisfy the following relations:
\begin{myEnumerate}
\item $\bar l_i^{1^{(k+1)}}\!\! = \bar l_i^{1^{(k')}}\!\! + \triangle\lambda_i^{(k+1)} ,\forall i\in\mathcal{V}_1\!\cap\!\mathcal{V}_2  $.
\item $\bar l_i^{2^{(k+1)}}\!\! = \bar l_i^{2^{(k')}}\!\! - \triangle\lambda_i^{(k+1)} ,\forall i\in\mathcal{V}_1\!\cap\!\mathcal{V}_2  $.
\item all the other coefficients remain unchanged.
\end{myEnumerate}
where $\triangle\lambda_i^{(k+1)} = \lambda_i^{(k+1)} - \lambda_i^{(k)}, \forall i\in\mathcal{V}_1\!\cap\!\mathcal{V}_2 $.
\end{prop}
\begin{proof}
    See Appendix \ref{app:update}.
\end{proof}

\begin{prop}
\label{prop:merge}
The coefficients of $f^{G'}(\mathbf{x}_{\mathcal{V}})$, which is the multi-linear polynomial of the merged graph $G'$, have the following relations with those of $f^{G_1^{(k+1)}}(\mathbf{x}_{\mathcal{V}_1})$ and $f^{G_2^{(k+1)}}(\mathbf{x}_{\mathcal{V}_2})$:
\begin{myEnumerate}
\item $l' =  l^{1^{(k+1)}}\!\! + l^{2^{(k+1)}}$ .
\item $\widehat l_i' =  \widehat l^{1^{(k+1)}}_i \!\! , \forall i\in\mathcal{V}_1\!\setminus\! \mathcal{V}_2;\quad$ $\widehat l_{ij}' = \widehat l_{ij}^{1^{(k+1)}}\!\!  ,\forall i\!\text{ \it or }\!j\in\mathcal{V}_1\!\setminus\!\mathcal{V}_2$ .
\item $\widetriangle l_i' = \widetriangle l^{2^{(k+1)}}_i \!\! , \forall i\in\mathcal{V}_2\!\setminus\!\mathcal{V}_1;\quad$  $ \widetriangle l_{ij}'=\widetriangle l_{ij}^{2^{(k+1)}} \!\!,\forall i\!\text{ \it or }\!j\in\mathcal{V}_2\!\setminus\!\mathcal{V}_1$ .
\item $\bar l_i' = \bar l^{1^{(k+1)}}_i \!\!+ \bar l^{2^{(k+1)}}_i  , \forall i\in\mathcal{V}_1\!\cap\!\mathcal{V}_2;\quad$  $\bar l_{ij}' = \bar l_{ij}^{1^{(k+1)}}\!\! + \bar l_{ij}^{2^{(k+1)}} \!\!, \forall i,j\in\mathcal{V}_1\!\cap\!\mathcal{V}_2$ .
\end{myEnumerate}
\end{prop}
\begin{proof}
    See Appendix \ref{app:merge}.
\end{proof}

\subsection{Correctness and efficiency of the merging}

The following theorem states the relations between the multi-linear polynomials of the original graph $G$ and that of the merged graph $G'$ shown in Fig. \ref{fig:twoSubgraphs}.
\begin{thm}
    \label{thm:conclusion}
    The multi-linear polynomials $f^{G'}(\mathbf{x}_\mathcal{V})$ of the merged graph $G'$ satisfies the following equation:
    \begin{equation}
	\label{eq:relation}
	f^{G}(\mathbf{x}_\mathcal{V}) = f^{G'}(\mathbf{x}_\mathcal{V}) + \sum_{m=1}^{K}\mathrm{F_1^{(m)}} + \sum_{m=1}^{K}\mathrm{F_2^{(m)}},
    \end{equation}
    where $f^G(\mathbf{x}_\mathcal{V})$ is the multi-linear polynomials of the original graph $G$, $\mathrm{F_1^{(m)}}$ and $\mathrm{F_2^{(m)}}$ are the flows computed in the $m^{th}$ iteration on subgraphs $G^{(m)}_1$ and $G^{(m)}_2$, respectively, and $K$ is the number of iterations of the two subgraphs before merging.
\end{thm}
\begin{proof}
	The original graph $G$ was first split into two subgraphs. Then, after $K$ iterations of maxflow computation and updating, the two subgraphs $G_1^{(K+1)}$ and $G_2^{(K+1)}$ were obtained. The following relations among the coefficients of $f^{G_1^{(K+1)}}(\mathbf{x}_\mathcal{V})$ , $f^{G_2^{(K+1)}}(\mathbf{x}_\mathcal{V})$ and $f^G(\mathbf{x}_\mathcal{V})$ can be drawn from Proposition.\ref{prop:split} -- \ref{prop:update}:
	\begin{equation}
	\widehat l^{1^{(K+1)}}_i \!\!= \widehat l_i , \forall i\in\mathcal{V}_1\!\setminus\! \mathcal{V}_2;\,\,\,\,
	\widehat l_{ij}^{1^{(K+1)}}\!\! = \widehat l_{ij},\forall i\!\text{ \it or }\!j\in\mathcal{V}_1\!\setminus\!\mathcal{V}_2,
	\end{equation}
	\begin{equation}
	\widetriangle l^{2^{(K+1)}}_i \!\!= \widetriangle l_i , \forall i\in\mathcal{V}_2\!\setminus\!\mathcal{V}_1;\,\,\,\,
	\widetriangle l_{ij}^{2^{(K+1)}} \!\!= \widetriangle l_{ij},\forall i\!\text{ \it or }\!j\in\mathcal{V}_2\!\setminus\!\mathcal{V}_1,
	\end{equation}
	\begin{equation}
	\bar l_{ij}^{1^{(K+1)}}\!\! = \bar l_{ij}^{2^{(K+1)}}\!\! =\frac{1}{2} \bar l_{ij},\forall i,j\in\mathcal{V}_1\!\cap\!\mathcal{V}_2,
	\end{equation}
	\begin{equation}
	\bar l^{1^{(K+1)}}_i\!\!  =\frac{1}{2} \bar l_i  + \sum_{m=1}^{K}\triangle\lambda_i^{(m+1)} , \forall i\in\mathcal{V}_1\!\cap\!\mathcal{V}_2,
	\end{equation}
	\begin{equation}
	\bar l^{2^{(K+1)}}_i\!\!  =\frac{1}{2} \bar l_i  - \sum_{m=1}^{K}\triangle\lambda_i^{(m+1)} , \forall i\in\mathcal{V}_1\!\cap\!\mathcal{V}_2,
	\end{equation}
	\begin{equation}
	l^{1^{(K+1)}}\!\!  = l^{1^{(1)}}  - \sum_{m=1}^{K}\mathrm{F}_1^{(m)},
	\end{equation}
	\begin{equation}
	l^{2^{(K+1)}}\!\!  = l^{2^{(1)}}  - \sum_{m=1}^{K}\mathrm{F}_2^{(m)}.
	\end{equation}
	When the two subgraphs $G_1^{(K+1)}$ and $G_2^{(K+1)}$ are merged into $G'$, according to Proposition \ref{prop:merge}, the coefficients of $f^{G'}(\mathbf{x}_\mathcal{V})$ are:
	\begin{equation}
	\label{eq:constTerm}
	\begin{split}
	l' &=  l^{1^{(K+1)}} + l^{2^{(K+1)}}\\
	& = l^{1^{(1)}} + l^{2^{(1)}} - \sum_{m=1}^{K}\mathrm{F}_1^{(m)} - \sum_{m=1}^{K}\mathrm{F}_2^{(m)} \\
	& = l - \sum_{m=1}^{K}\mathrm{F}_1^{(m)} - \sum_{m=1}^{K}\mathrm{F}_2^{(m)} ,
	\end{split}
	\end{equation}
	\begin{equation}
	\begin{split}
	&{}\widehat l'_i=\widehat l^{1^{(K+1)}}_i \!\!= \widehat l_i , \forall i\in\mathcal{V}_1\!\setminus\! \mathcal{V}_2;\\
	&{}\widehat l'_{ij}= \widehat l_{ij}^{1^{(K+1)}}\!\! = \widehat l_{ij},\forall i\!\text{ \it or }\!j\in\mathcal{V}_1\!\setminus\!\mathcal{V}_2	.
        \end{split}
	\end{equation}
	\begin{equation}
        \begin{split}
	&{}\widetriangle l'_i=  \widetriangle l^{2^{(K+1)}}_i \!\!= \widetriangle l_i , \forall i\in\mathcal{V}_2\!\setminus\!\mathcal{V}_1;\\
	&{}\widetriangle l'_{ij}=\widetriangle l_{ij}^{2^{(K+1)}} \!\!= \widetriangle l_{ij},\forall i\!\text{ \it or }\!j\in\mathcal{V}_2\!\setminus\!\mathcal{V}_1 .
        \end{split}
	\end{equation}
	\begin{align*}
	\bar l'_i &=  \bar l^{1^{(K+1)}}_i\!\! + \bar l^{2^{(K+1)}}_i\\
	&  =\left(\frac{1}{2} \bar l_i  + \sum_{m=1}^{K}\triangle\lambda_i^{(m+1)}\right) +  \left(\frac{1}{2} \bar l_i  - \sum_{m=1}^{K}\triangle\lambda_i^{(m+1)}\right) \\
	& = \bar l_i,\quad \forall i\in\mathcal{V}_1\!\cap\!\mathcal{V}_2 \,\,\numberthis ,
	\end{align*}
	\begin{equation}
	\label{eq:pairwiseIntersection}
	\bar l'_{ij}=\bar l_{ij}^{1^{(K+1)}}\!\!+ \bar l_{ij}^{2^{(K+1)}}\!\! =\frac{1}{2} \bar l_{ij}+\frac{1}{2} \bar l_{ij} =  \bar l_{ij},\forall i,j\!\in\!\mathcal{V}_1\cap\mathcal{V}_2 . 
	\end{equation}
	Therefore, only the constant term of $f^{G'}(\mathbf{x}_\mathcal{V})$ differs from that of  $f^{G}(\mathbf{x}_\mathcal{V})$ by $\sum_{m=1}^{K}\mathrm{F}_1^{(m)} + \sum_{m=1}^{K}\mathrm{F}_2^{(m)}$. All the other terms have the same coefficients, and Theorem \ref{thm:conclusion} is proved.
\end{proof}

Note that $\sum_{m=1}^{K}\mathrm{F_1^{(m)}}\allowbreak + \sum_{m=1}^{K}\mathrm{F_2^{(m)}}$ is the accumulated flows in the subgraphs, which is independent of the variables $\mathbf{x}_\mathcal{V}$. Since only the optimal arguments $\accentset{*}{\mathbf{x}}_{\mathcal{V}}$ at which $f^G(\mathbf{x}_{\mathcal{V}})$ reaches its global minimum, not the global minimum itself, matters for a graph cuts problem defined on graph $G$, if the pseudo-boolean functions of two graph cuts problems only differ by flows which do not depend on their arguments $\mathbf{x}_\mathcal{V}$, the two graph cuts problems are equivalent. And by the convention of the combinatorial optimization community, this difference is often referred to as a constant. An illustrative example is given in Appendix \ref{app:illustrative}. Therefore, Theorem \ref{thm:conclusion} reveals the following two important things: 1) The two graph cuts problems, which are defined on the original graph $G$ and the merged graph $G'$, are equivalent because their corresponding pseudo-boolean functions only differ by a constant. Hence the same global optimal solution can be obtained from the merged graph $G'$. 2) Moreover, this constant is the summation of all the flows computed in the two subgraphs in all the $K$ iterations, indicating that the flows are reused when computing the maxflow for the merged graph $G'$. As a result, computing the maxflow for the merged graph $G'$ will be much faster than computing the maxflow for the original graph $G$ from scratch.

Theorem \ref{thm:conclusion} also gives a good explanation of the parallel BK-algorithm from the perspective of flows. The algorithm can be regarded as decomposing the maximum flow computation of the original graph into subgraphs. Once the summation of all the accumulated flows in the subgraphs reaches the maximum flow of the original graph, the parallel BK-algorithm converges. Therefore, whenever the subgraphs can easily accumulate flows, the parallel BK-algorithm can easily converge.

To evaluate the effectiveness of reusing flows in the proposed merging method, we once again apply the two fore-/back-ground image segmentation methods (seg1 and seg2 used in section \ref{sec:convergence}) on all 500 images of the Berkeley segmentation dataset~\cite{amfm_pami2011}. Here, the naive converged parallel BK-algorithm (Algorithm \ref{alg:naiveConvergedParallel}), whose $N$ is fixed at $2$ and whose $\mathrm{ITER}$ is set to $1,5,15$ and $20$, respectively, is used. Fig.~\ref{fig:relativeTimeFlowMerge} shows the distribution of the relative times, which are defined as the ratio of the times for computing the maxflow of the merged graph $G'$ to that of the original graph $G$. The figure also shows the relative used flows, which are defined as the ratio of all the accumulated flows in the subgraphs before merging, that is $\sum_{m=1}^{K}\mathrm{F}_1^{(m)}+\sum_{m=1}^{K}\mathrm{F}_2^{(m)}$ in (\ref{eq:relation}), to the maxflows of the original graph $G$.
\begin{figure*}[htbp]
    \begin{center}
	\subfloat[Relative times for seg1]{\label{fig:relativeTimeMergeSeg1}\includegraphics[width=0.45\linewidth]{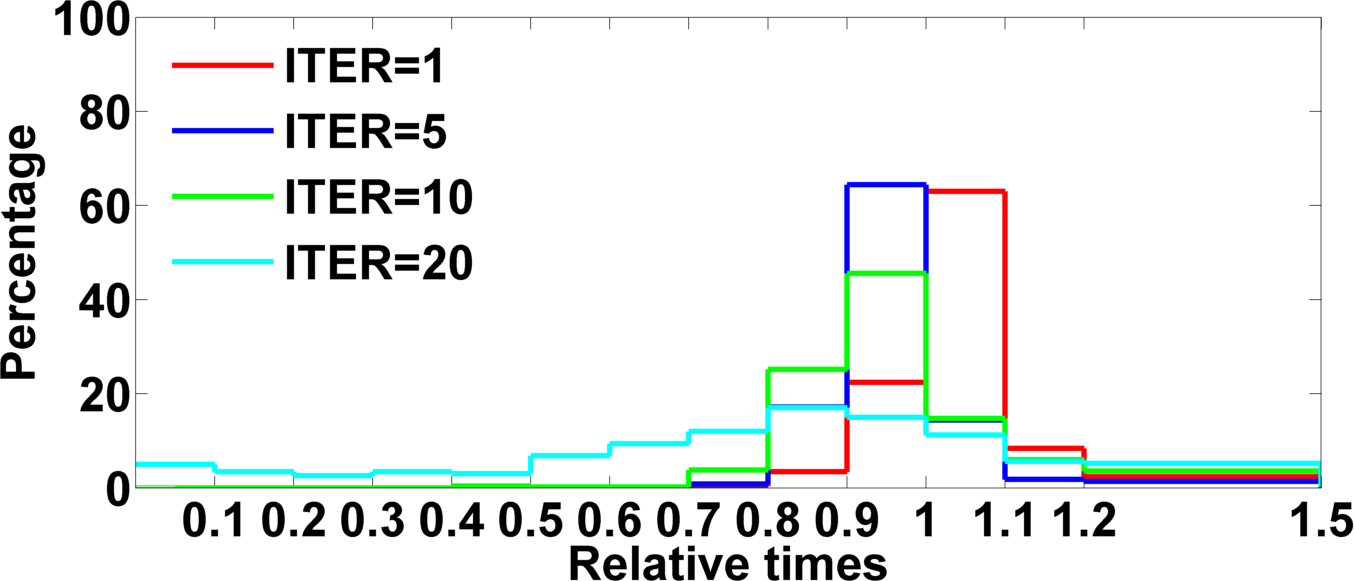}}\qquad\qquad
	\subfloat[Relative reused flows for seg1]{\label{fig:relativeFlowMergeSeg1}\includegraphics[width=0.45\linewidth]{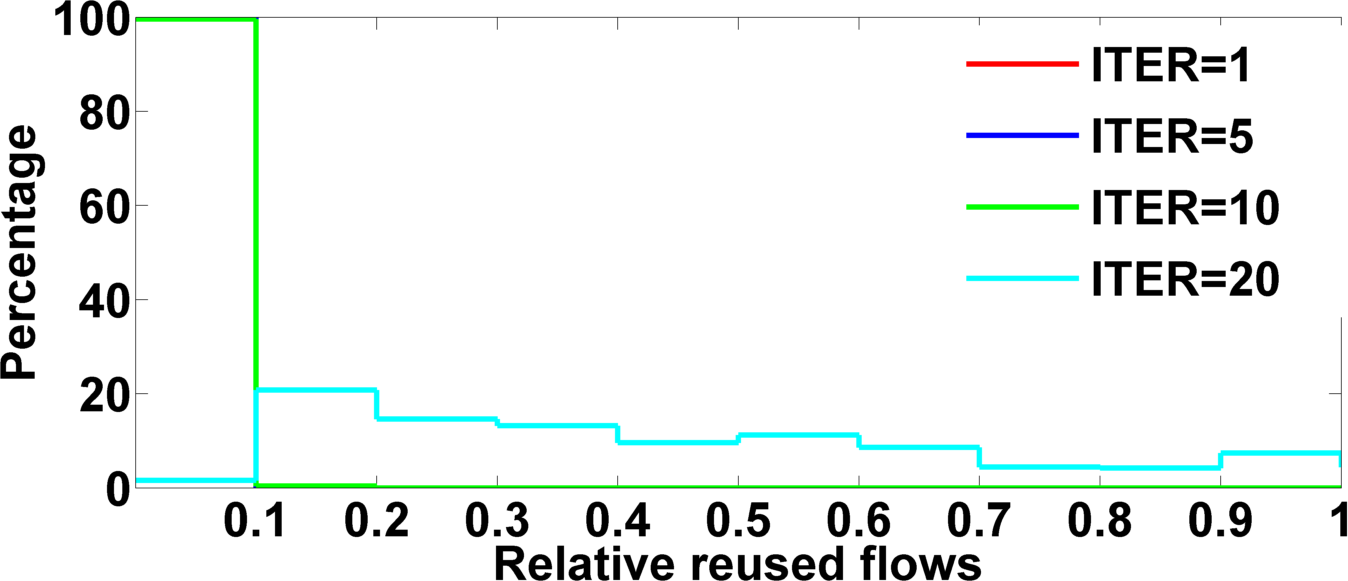}}\\
	\subfloat[Relative times for seg2]{\label{fig:relativeTimeMergeSeg2}\includegraphics[width=0.45\linewidth]{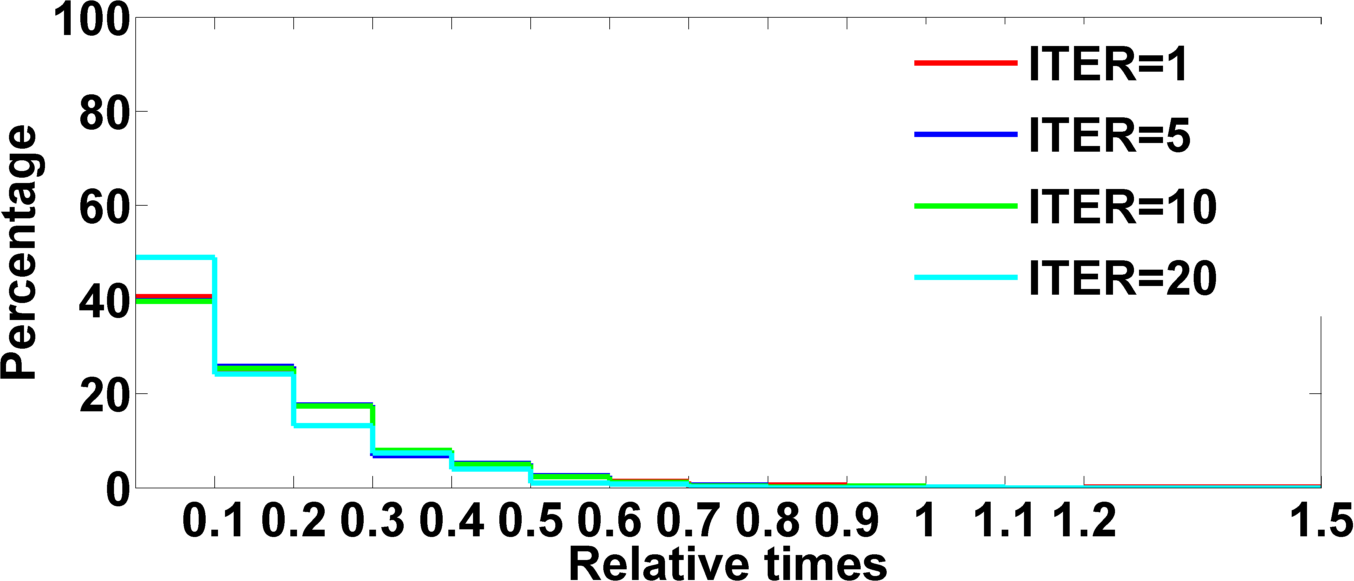}}\qquad\qquad
	\subfloat[Relative reused flows for seg2]{\label{fig:relativeFlowMergeSeg2}\includegraphics[width=0.45\linewidth]{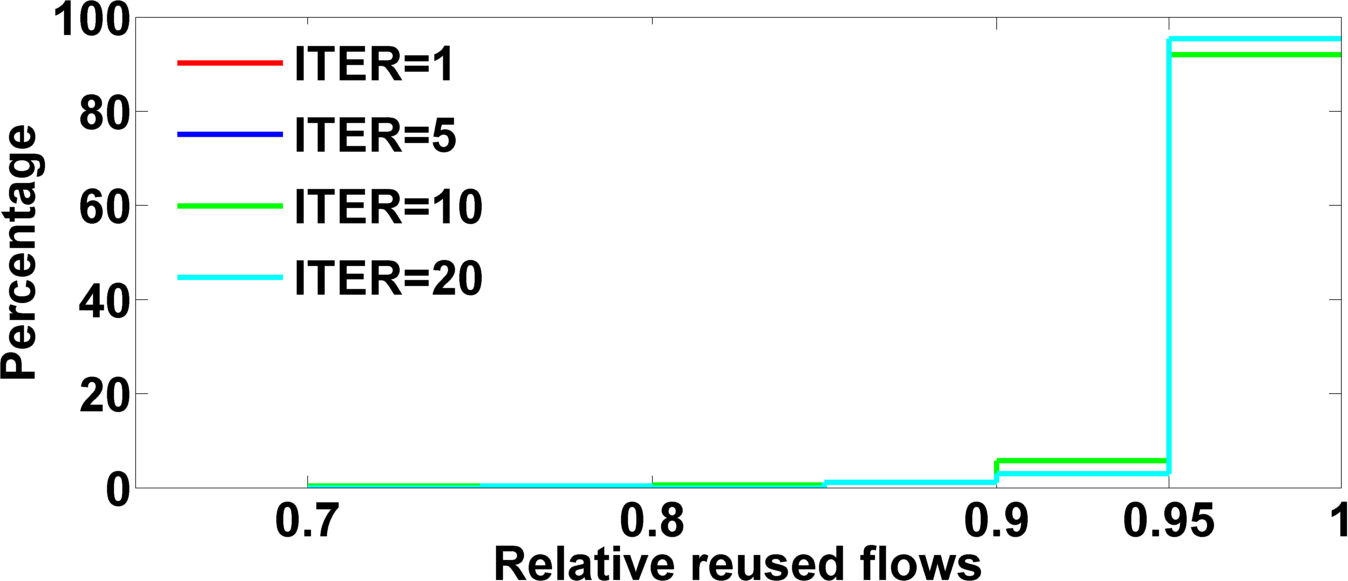}}
    \end{center}
    \caption{Relative times and relative reused flows for the two segmentations. In all the figures, $\mathrm{ITER}$ is the iteration control parameter used in Algorithm \ref{alg:naiveConvergedParallel}.}
    \label{fig:relativeTimeFlowMerge}
\end{figure*}

The relative times for seg1 when $\mathrm{ITER}=1,5,10$, as shown in Fig. \subref*{fig:relativeTimeMergeSeg1}, are roughly equal to 1, which means there is no speedup in computing the merged graph $G'$. By examining Fig. \subref*{fig:relativeFlowMergeSeg1}, we found that the reused flows in these three cases are very small (their histograms are all concentrated in the first bin). It is precisely this no-flow to reuse that accounts for the ineffectiveness in reusing flows. Because for the original graph $G$ of seg1 only the pixels in the leftmost column are connected to the source and the pixels in the rightmost column are connected to the sink,  the maxflows of the two subgraphs in the first iteration are all equal to $0$, i.e., $\mathrm{F_1^{(1)}}=\mathrm{F_2^{(1)}}=0$, because none of the pixels of $G_1^{(1)}$ are connected to the sink and none of the pixels of $G_2^{(1)}$ are connected to the source. Once the subgraphs accumulated a certain amount of flows, a sustained acceleration in the computations of the merged graph occurred, as shown in Fig. \ref{fig:relativeTimeFlowMerge}\subref{fig:relativeTimeMergeSeg1},\subref{fig:relativeFlowMergeSeg1} for the $\mathrm{ITER}=20$ case. 

The flows in the subgraphs for seg2 are fairly easy to accumulate, as can be seen in Fig. \subref*{fig:relativeFlowMergeSeg2}, where even in the $\mathrm{ITER}=1$ case, over $90\%$ of the images have accumulated over $95\%$ of the flows of the original graph. This leads to a significant speedup in the merged graph, where the maxflow computations for approximately $40\%$ of the merged graphs are over 10 times faster than those of the original graphs, and over $90\%$ of the images could achieve a speedup factor of at least 3. Although further increase of the flows in the subgraphs is relatively slow, once this increase becomes apparent, it will be reflected in the improvement of the relative times as well, as shown in Fig. \ref{fig:relativeTimeFlowMerge}\subref{fig:relativeTimeMergeSeg2},\subref{fig:relativeFlowMergeSeg2} for the $\mathrm{ITER}=20$ case. 

The experiments here show that whenever there are accumulated flows in the subgraphs, these flows will be reused effectively in the merged graph. It is worth noting that the correctness and efficitiveness of the merging method can be easily generalized to cases of more than two subgraphs.

\section{Dynamic parallel graph cuts algorithm}
With the two key building blocks -- splitting and merging at hand, our dynamic parallel graph cuts algorithm is described in Algorithm \ref{alg:dynamicParallel}.

\begin{algorithm}[htbp]                      
\caption{Dynamic Parallel Graph Cuts Algorithm}          
\label{alg:dynamicParallel}                           

\begin{algorithmic}[1]                    
	\STATE \textbf{{\bf Set} $nDiff:=+\infty$;}
	\STATE {Initial splitting;} 
	\WHILE{$nDiff > 0$}
	\STATE \emph{A subgraph can be further split into a number of subgraphs}; \label{code:splitting}
	\REPEAT
	\STATE \emph{Compute maxflow for all the subgraphs;} \label{code:maxflow}
	\STATE Update all the subgraphs accordingly; \label{code:update}
	\STATE Count the number of nodes that disagree on their optimal values, denoted as $nDiff$;\label{code:nDiff}
	\UNTIL{$nDiff == 0$ \OR \emph{some other conditions are met}} \label{code:stoppingRule}
	\STATE \emph{Any number of neighboring subgraphs can be merged;} \label{code:merging}
	\ENDWHILE
\end{algorithmic}
\end{algorithm}

In essence, Algorithm \ref{alg:dynamicParallel} is a general parallel framework, where lines \ref{code:splitting}, \ref{code:maxflow}, \ref{code:stoppingRule} and \ref{code:merging} do not specify their concrete realization but give the user complete freedom to choose their own implementations for specific problem and specific parallel and distributed platforms.  
Line \ref{code:maxflow} specifies the used maxflow solver. Lines \ref{code:splitting}, \ref{code:stoppingRule} and \ref{code:merging} jointly determine when and how the subgraphs are merged or further split: line \ref{code:stoppingRule} specifies the time for merging and its realization is given in line \ref{code:merging}, and line \ref{code:splitting} determines the dynamic splitting strategy.

It can be seen from lines \ref{code:maxflow} -- \ref{code:nDiff} of Algorithm \ref{alg:dynamicParallel} that a series of similar maxflow problems (more specifically, graph cuts problems in adjacent iterations only differing in a tiny fraction of t-link values) need to be solved in order to obtain the global optimal solutions of the original graph. Therefore, any graph cuts algorithm that is efficient in this dynamic setting can be used as the maxflow solver in line \ref{code:maxflow}. In other words, any such serial graph cuts algorithm can be parallelized using Algorithm \ref{alg:dynamicParallel} to further boost its efficiency and scalability.

Since the BK-algorithm~\cite{boykov2004experimental} has some efficient dynamic variants~\cite{kohli2007dynamic,kohli2005efficiently,juan2006active}, it is parallelized (used as the maxflow solver) by the parallel BK-algorithm~\cite{strandmark2010parallel,strandmark2011parallel}. Pseudoflow-based maxflow algorithms~\cite{hochbaum2008pseudoflow}, which can often be efficiently warm started after solving a similar problem, can also be parallelized using Algorithm \ref{alg:dynamicParallel}. For example, the most recently proposed Excesses Incremental Breadth-First Search (Excesses IBFS) algorithm~\cite{goldberg2015faster}, which is a generalization of the incremental breadth-first search (IBFS) algorithm~\cite{goldberg2011maximum} to maintain a pseudoflow, could act as a better candidate maxflow solver in line \ref{code:maxflow} of Algorithm \ref{alg:dynamicParallel}. The reason is that the performance of the Excesses IBFS is competitive (usually faster) compared with other dynamic graph cuts algorithms. More importantly, it has a polynomial running time guarantee whereas there is no known polynomial time limit for the BK-algorithm.

Dynamic merging aims to dynamically remedy the over-splitting during the running time, and it can make the dynamic parallel graph cuts algorithm converge to the globally optimal solutions within a predefined number of iterations. Here is a general procedure: 1) the graph is initially split into $N$ subgraphs, and no further splitting is permitted; 2) invoke the merging operations every $K$ iterations; 3) for a merging operation, every neighboring $\ell$ subgraphs are merged into a single graph. In the worst case, the merging operations can be invoked $\log_\ell^N$ times, leaving only one graph, from which the global optimal solutions can be obtained. Therefore, the dynamic parallel graph cuts algorithm can converge within no more than $K\log_\ell^N$ iterations under these settings. By setting $K=1$ and $\ell=N$, only one iteration is needed in this case.

Dynamic splitting is also important in some cases. Since the connection strengths may vary dramatically even in a uniformly connected graph, it is difficult to design an optimal scheme to divide the graph into subgraphs at the beginning. In addition, since the commonly used equal division strategy in most cases can only balance the memory storage, not the computational load, and because the situation could be further worsened by the dynamic merging process, dynamic splitting is needed for dynamic workload balancing.

Incorporating dynamic merging and dynamic splitting enables the parallel granularity, load balance and configuration of the subgraph splitting to be dynamically adjusted during the parallel computation. Hence Algorithm \ref{alg:dynamicParallel} is here called dynamic. Our dynamic parallel graph cuts algorithm is a general framework that could be tailored to fit different types of graphs and platforms for better efficiency and convergence rate.

To summarize, our proposed dynamic parallel graph cuts algorithm has the following three main advantages: 1) any serial graph cuts solver that is efficient in the dynamic setting can be parallelized using the dynamic parallel graph cuts algorithm; 2) the dynamic parallel graph cuts algorithm guarantees convergence, moreover, it can converge to the global optimum within a predefined number of iterations; 3) the subgraph splitting configuration can be adjusted freely during the parallelization process.

\section{Experiments}
Although advanced splitting and merging strategies could be cleverly used to adjust the subgraph splitting configuration during the whole parallelization process, finding such a general strategy is not an easy task. For a fair comparison with the parallel BK-algorithm, only the BK-algorithm is used as the graph cuts solver in the iterations of all the parallel graph cuts algorithms. Therefore, in all the following experiments, only the naive converged parallel BK-algorithm (Algorithm \ref{alg:naiveConvergedParallel}) is used, in comparison with the parallel BK-algorithm, to evaluate the effectiveness of our introduced merging method. There are only two parameters for all these experiments: $\mathrm{ITER}$ controls the starting time of the merging operations in the naive converged parallel BK-algorithm, and $\mathrm{TRDS}$ is the number of the used computational threads for the two parallel algorithms.

\subsection{Image segmentation}
The performance of the naive converged parallel BK-algorithm is first evaluated on the two image segmentation problems, seg1 and seg2, on all 500 images of the Berkeley segmentation dataset~\cite{amfm_pami2011}, which are used in section \ref{sec:convergence} to evaluate the convergence problem in the parallel BK-algorithm. $\mathrm{TRDS}$ is set to $4$ and $8$, respectively. And for each settings of $\mathrm{TRDS}$, $\mathrm{ITER}$ is set to $15,20,25,30$, respectively. We further define three parameters, $t_\mathrm{BK}, t_\mathrm{PBK} \text{ and } t_\mathrm{CPBK}$, to represent the time elapsed for maxflow calculations using the BK-algorithm, the parallel BK-algorithm and the naive converged parallel BK-algorithm, respectively. This image dataset can be divided into two parts according to whether the images can converge or not in the parallel BK-algorithm -- referred to as successImages and failedImages, respectively.  

\begin{table}
    \caption{Number of images that have triggered merging operations for seg2.}
    \label{tab:numMerge}
    \begin{center}
	\begin{tabular}{l|*{4}{p{0.40cm}}|c}
	    \toprule
	    \diagbox{\footnotesize $\mathrm{TRDS}$}{\footnotesize $\mathrm{ITER}$} & 15 & 20 & 25 & 30 & successImages\\
	    \midrule
	    4 & 290 & 89 & 49 & 40 & 486 \\
	    8 & 197 & 91 & 66 & 54 & 362 \\
	    \bottomrule
	\end{tabular}
    \end{center}
\end{table}

\begin{figure*}[htbp]
	\begin{center}
		\subfloat[\scriptsize Relative times for seg1 on failedImages of 4 computational threads. The red, blue, green and cyan lines represent $t_\mathrm{CPBK}/t_\mathrm{BK}$ with $\mathrm{ITER}$ set to $15$, $20$, $25$ and $30$, respectively, and their medians are $1.58$, $1.23$, $1.34$ and $1.17$.]{\label{fig:improper4}\includegraphics[width=0.31\linewidth]{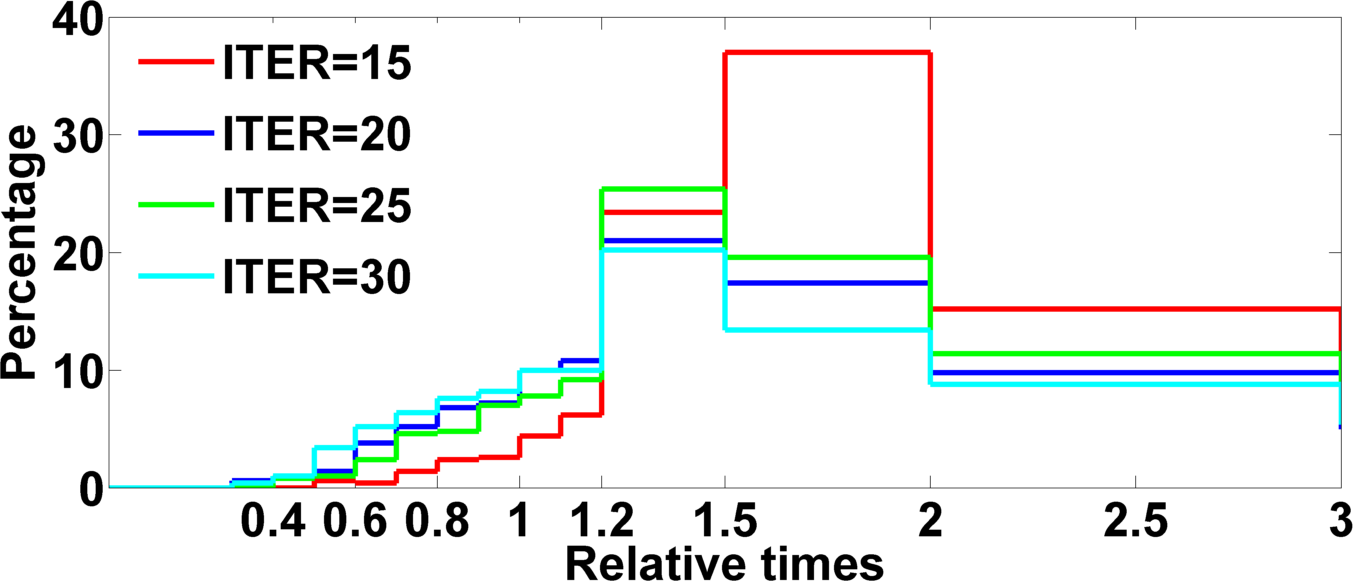}}\quad
		\subfloat[\scriptsize Relative times for seg1 on failedImages of 8 computational threads. The red, blue, green and cyan lines represent $t_\mathrm{CPBK}/t_\mathrm{BK}$ with $\mathrm{ITER}$ set to $15$, $20$, $25$ and $30$, respectively, and their medians are $1.89$, $1.74$, $1.91$ and $2.03$.]{\label{fig:improper8}\includegraphics[width=0.31\linewidth]{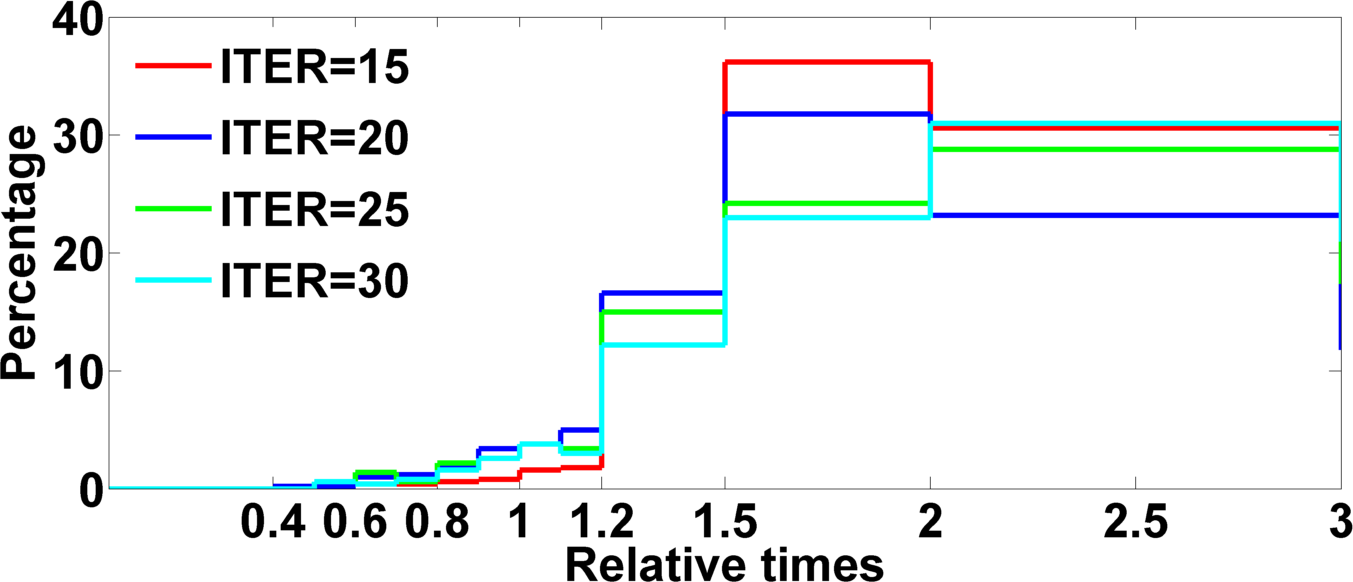}}\quad
		\subfloat[\scriptsize Relative times for seg2 on successImages of $4$ computational threads. The ORG in black represents $t_\mathrm{PBK}/t_\mathrm{BK}$, and its median is $0.56$. The red, blue, green and cyan lines represent $t_\mathrm{CPBK}/t_\mathrm{BK}$ with ITER set to $15$, $20$, $25$ and $30$, respectively, and their medians are $0.68$, $0.58$, $0.57$ and $0.55$.]{\label{fig:properSuccess4}\includegraphics[width=0.31\linewidth]{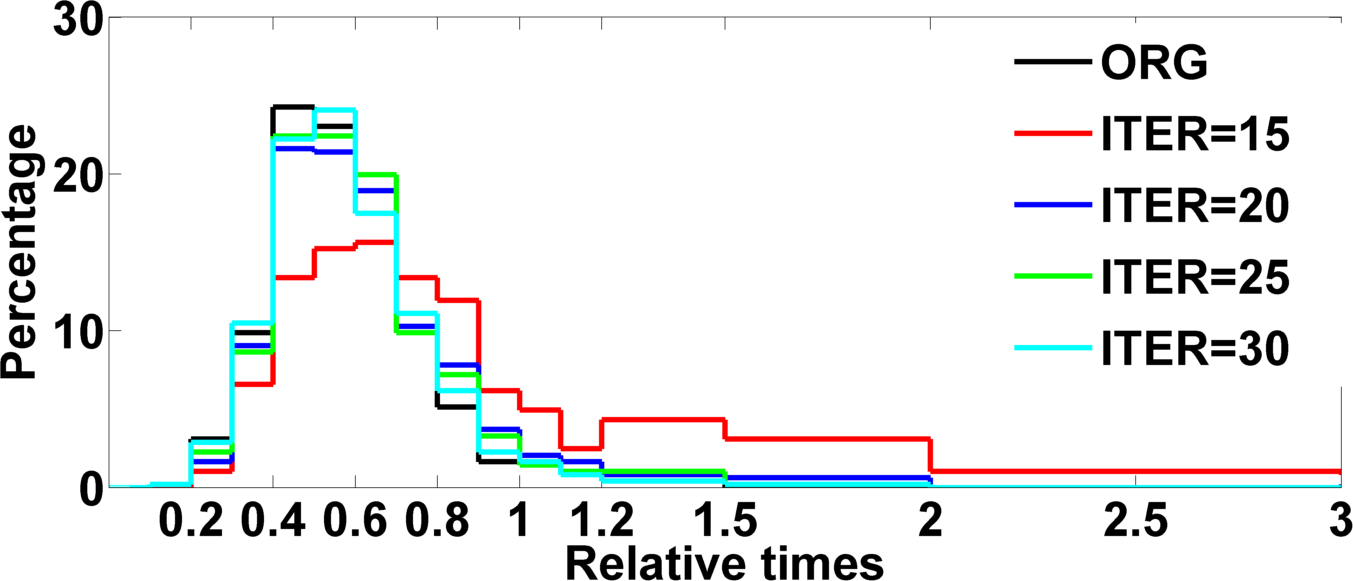}}\\
		\subfloat[\scriptsize Relative times for seg2 on successImages of $8$ computational threads. The ORG in black represents $t_\mathrm{PBK}/t_\mathrm{BK}$, and its median is $0.52$. The red, blue, green and cyan lines represent $t_\mathrm{CPBK}/t_\mathrm{BK}$ with ITER set to $15$, $20$, $25$ and $30$, respectively, and their medians are $0.56$, $0.50$, $0.50$ and $0.49$.]{\label{fig:properSuccess8}\includegraphics[width=0.31\linewidth]{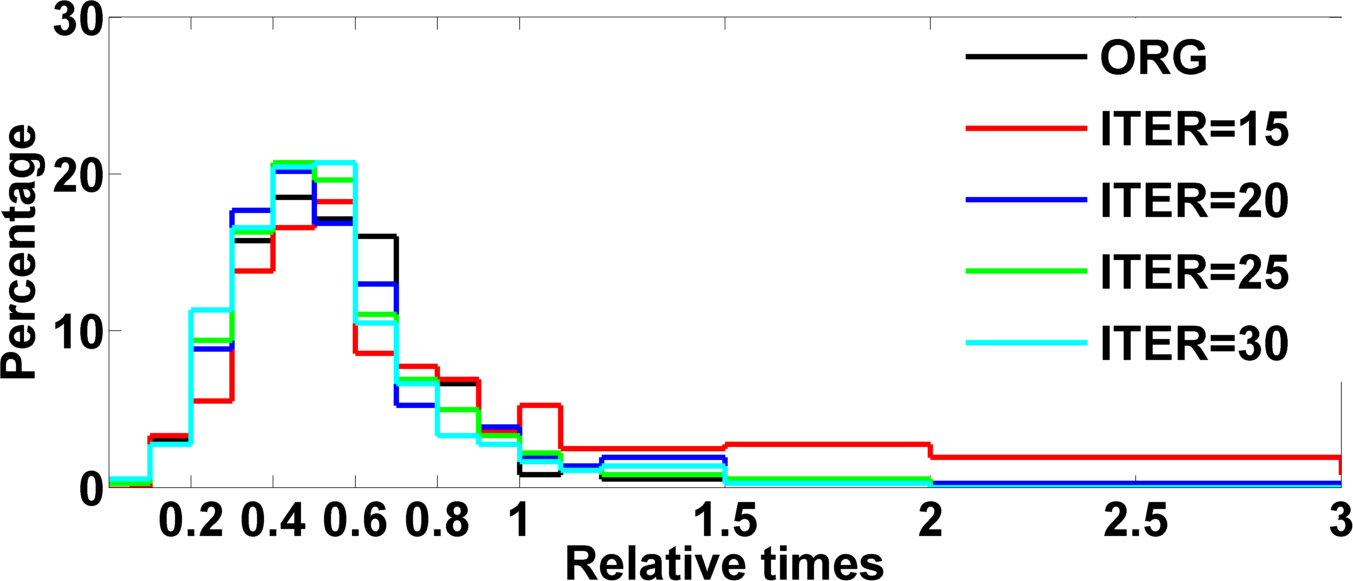}}\quad
		\subfloat[\scriptsize Relative times for seg2 on all 14 failedImages of 4 computational threads. The red, blue, green and cyan lines represent $t_\mathrm{CPBK}/t_\mathrm{BK}$ with $\mathrm{ITER}$ set to $15$, $20$, $25$ and $30$, respectively, and their medians are $0.75$, $0.71$, $0.73$ and $0.70$.]{\label{fig:properFailure4}\includegraphics[width=0.31\linewidth]{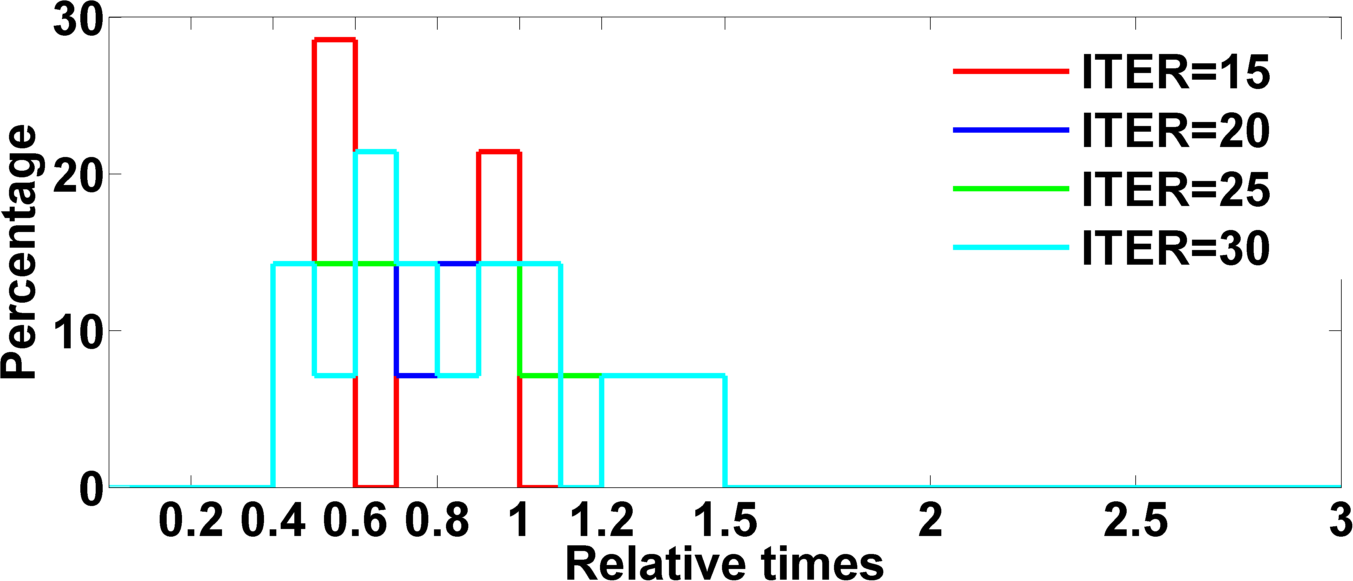}}\quad
		\subfloat[\scriptsize Relative times for seg2 on all 138 failedImages of 8 computational threads. The red, blue, green and cyan lines represent $t_\mathrm{CPBK}/t_\mathrm{BK}$ with $\mathrm{ITER}$ set to $15$, $20$, $25$ and $30$, respectively, and their medians are $0.66$, $0.69$, $0.72$ and $0.73$.]{\label{fig:properFailure8}\includegraphics[width=0.31\linewidth]{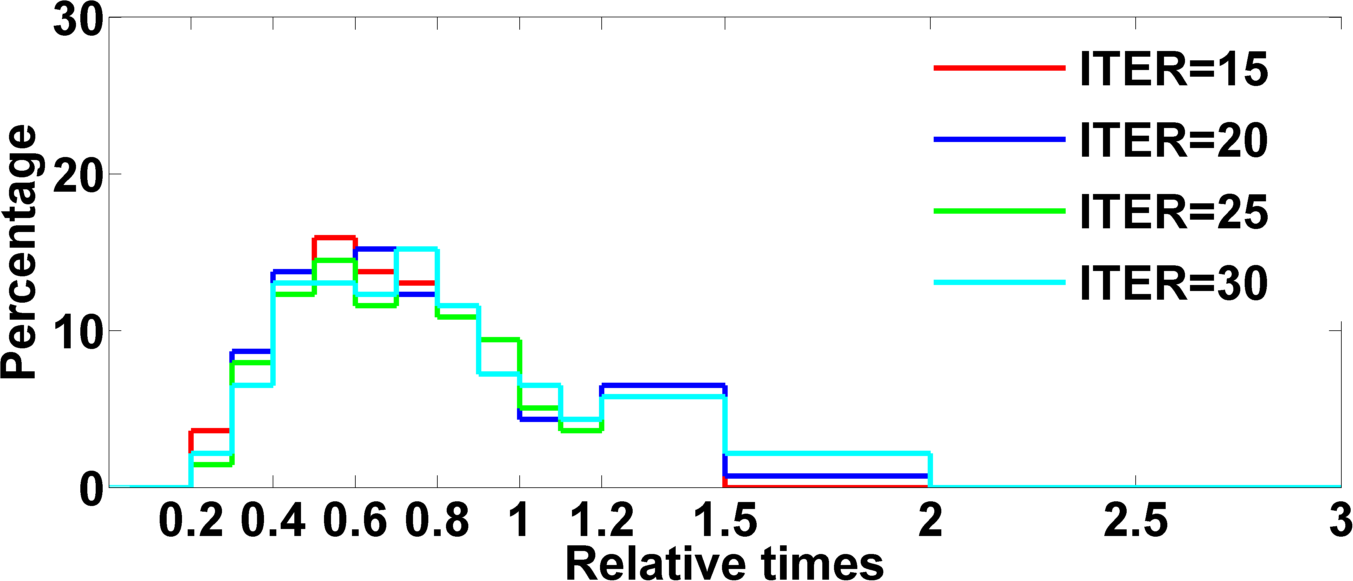}}	
	\end{center}
	\caption{Relative times for the two segmentation methods}
	\label{fig:relativeTimeSeg}
\end{figure*}

The relative times ($t_\mathrm{CPBK}/t_\mathrm{BK}$) for seg1 on the failedImages, which contains all the images of the dataset in these cases under $4$ and $8$ computational threads, are shown in Fig. \subref*{fig:improper4} and Fig. \subref*{fig:improper8}, respectively. Although all these images that failed to converge in the parallel BK-algorithm can certainly converge in the naive converged parallel BK-algorithm, they usually take a little more time than in the serial BK-algorithm, as shown in Fig. \ref{fig:relativeTimeSeg}\subref{fig:improper4},\subref{fig:improper8}. The improper settings of $\mathrm{ITER}$, \eg, $\mathrm{ITER}=15$, will degrade the performance of the naive converged parallel BK-algorithm. This is possibly because the subgraphs could not accumulate enough flows for reuse, and the merging operations would be invoked too many times, more than the necessary, if $\mathrm{ITER}$ is too small. However, the naive converged parallel BK-algorithm can come up with similar results with a wide range of $\mathrm{ITER}$, \eg, $\mathrm{ITER}=20,25,30$.

For seg2, successImages accounts for the majority of the $500$ images. Table \ref{tab:numMerge} shows the number of images that belong to successImages but that also triggered the merging operations in the naive converged parallel BK-algorithm under different $\mathrm{ITER}$ and $\mathrm{TRDS}$ settings. It is evident from Table \ref{tab:numMerge} that the small $\mathrm{ITER}$ values will lead to a large number of images to invoke the merging operations.
Experiments are carried out on successImages and failedImages, as shown in Fig. \ref{fig:relativeTimeSeg}\subref{fig:properSuccess4}-\subref{fig:properSuccess8}, and Fig. \ref{fig:relativeTimeSeg}\subref{fig:properFailure4}-\subref{fig:properFailure8}, respectively. The relative times ($t_\mathrm{PBK}/t_\mathrm{BK}$) of the parallel BK-algorithm on successImages are also shown in Fig. \ref{fig:relativeTimeSeg}\subref{fig:properSuccess4}-\subref{fig:properSuccess8}, shown as ORG in black, for comparison. It can be seen from Fig. \ref{fig:relativeTimeSeg}\subref{fig:properSuccess4}-\subref{fig:properSuccess8} that except for the $\mathrm{ITER}=15$ case, which is slightly inferior to the parallel BK-algorithm, the naive converged parallel BK-algorithm performs equally as well as the parallel BK-algorithm under all the other $\mathrm{ITER}$ and $\mathrm{TRDS}$ settings. Moreover, the naive converged parallel BK-algorithm is approximately $30\%$ faster, on average, than the serial BK-algorithm on failedImages, as shown in Fig. \ref{fig:relativeTimeSeg}\subref{fig:properFailure4}-\subref{fig:properFailure8}. This is a promising result. It shows that the naive converged parallel BK-algorithm not only has convergence guarantee but may also have a considerable speedup even for these images for which the parallel BK-algorithm failed to converge. Concerning the parameter selection of the naive converged parallel BK-algorithm, because the algorithm is not sensitive to the settings of $\mathrm{ITER}$, it can be set to a number in the range $15-30$, and $\mathrm{TRDS}$ can be set to a number in the range $2-8$. Usually, the larger $\mathrm{TRDS}$ is, the faster the naive converged parallel BK-algorithm is. However, the speedup gain grows very slowly if $\mathrm{TRDS}$ is further increased. 

\subsection{Semantic segmentation}
Among many other multi-label problems in computer vision~\cite{scienceChina_cite2}, where maximum flow solvers are invoked many times as a subroutine for each problem instance, such as stereo, image denoising, semantic segmentation might be one of the most challenging applications for the two parallelized BK-algorithms because of the following two difficulties: 1) typically, only a few labels (object classes) can appear in an image, but the label set is relatively large, so as to contain all the possible labels that can occur in the whole image set; 2) each label that appears is usually concentrated on a small part of the image except for the background. Due to the first difficulty, very few flows can be pushed from source to sink for the graphs built for the move-making energies that moved to the non-existent labels of the image, which makes the overhead much more significant for creating, synchronizing and terminating multiple threads. The latter difficulty makes the workload very unbalanced among the computational threads. We used the semantic segmentation algorithm proposed by Shotton~\cite{shotton2006textonboost} on the MSRC-21 dataset~\cite{shotton2006textonboost}, where the unary potentials for pixels are the negative log-likelihood of the corresponding output of the classifiers derived from TextonBoost~\cite{shotton2006textonboost}. The most discriminative weak classifiers are found using multi-class Gentle Ada-Boost~\cite{torralba2004sharing}, and the pairwise potentials are the contrast sensitive potentials~\cite{boykov2001iccv} defined on the eight-neighborhood system. The energy functions are solved via $\alpha$-expansion~\cite{boykov2001fast}. 

\begin{figure*}[htbp]
	\begin{center}
		\subfloat[\scriptsize Relative times ($t_\mathrm{PBK}/t_\mathrm{BK}$) for semantic segmentation on all the images of the MSRC-21 dataset that can converge in the parallel BK-algorithm with $\mathrm{TRDS}$ set to 2, 4, 8 and 12. The medians are $1.32$, $1.14$, $0.97$ and $1.29$.]{\label{fig:semanticORG}\includegraphics[width=0.45\linewidth]{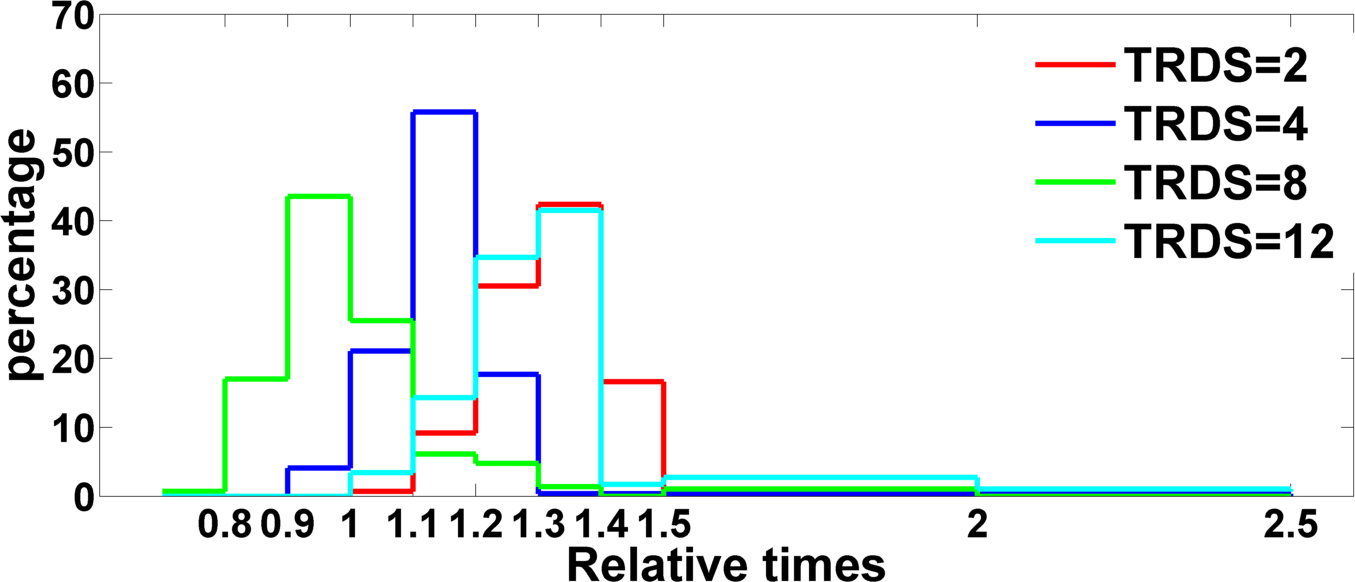}}\qquad\qquad
		\subfloat[\scriptsize Relative times ($t_\mathrm{CPBK}/t_\mathrm{PBK}$) on the unnecessary merged images in the naive converged parallel BK-algorithm with $\mathrm{TRDS}$ set to $8$ and $12$ and $\mathrm{ITER}$ set to $15$ and $25$. The medians are $1.01$, $1.00$, $0.96$ and $0.96$.]{\label{fig:semanticPGCMerged}\includegraphics[width=0.45\linewidth]{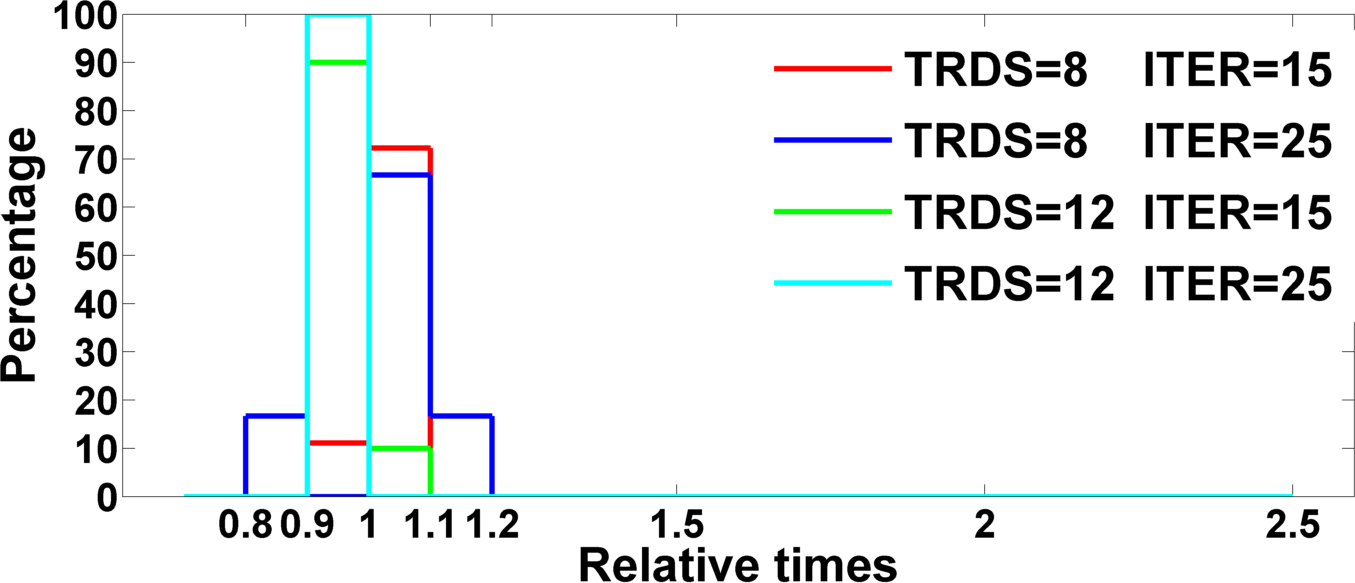}}
	\end{center}
	\caption{Relative times for semantic segmentation on the MSRC-21 dataset}
	\label{fig:relativeTimeMsrc}
\end{figure*}

Fig. \subref*{fig:semanticORG} shows the relative times for the parallel BK-algorithm using 2, 4, 8 and 12 computational threads on all the test images of the MSRC-21 dataset. It can be seen that the parallel BK-algorithm is always slower than the serial BK-algorithm, except for the 8 computational threads case, in which the parallel BK-algorithm is a bit faster than the BK-algorithm. Although the acceleration is poor, only one out of the $295$ test images is non-convergent for the 4, 8 and 12 computational threads cases. Moreover, only when $\alpha = 16$, which corresponds to solving the `road'-expansion energy function, does the parallel BK-algorithm failed to converge. This non-convergent example is shown in Fig. \ref{fig:failureExample}, where Fig.~\subref*{fig:sgIMG} is the original image and Fig.~\subref*{fig:sgGT} is the hand-labeled ground truth. The different pixel colors represent different object classes, whose correspondence is shown in Fig.~\subref*{fig:sgLegend}. The black pixels are the `void' pixels, which are ignored when computing semantic segmentation accuracies and are treated as unlabeled pixels during training. The semantic segmentation result when the serial BK-algorithm is used as the graph cuts solver is shown in Fig.~\subref*{fig:sgBK}. It is the benchmark of the correctness for all the other parallel graph cuts algorithms, and its running time acts as the baseline. The semantic segmentations by the parallel BK-algorithm with 4 and 8 computational threads are shown in Fig.~\subref*{fig:sg4PBK} and Fig.~\subref*{fig:sg8PBK} respectively. Fig.~\subref*{fig:sgCPBK} shows the result of the naive converged parallel BK-algorithm with 8 computational threads. The dotted blue lines in Fig.~\ref{fig:failureExample}\subref{fig:sg4PBK}--\subref{fig:sgCPBK} are the boundaries of the neighboring subgraphs. In contrast to Fig.~\ref{fig:diffNumThreads}, the graph is horizontally split for parallel graph cuts computation.

\begin{figure*}[htbp]
    \begin{center}
	\subfloat[Original image]{\label{fig:sgIMG}\includegraphics[width=0.32\linewidth]{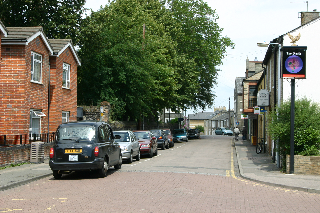}}\,
	\subfloat[Ground truth]{\label{fig:sgGT}\includegraphics[width=0.32\linewidth]{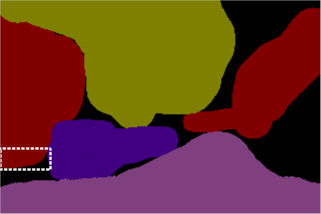}}\,
	\subfloat[The BK-algorithm]{\label{fig:sgBK}\includegraphics[width=0.32\linewidth]{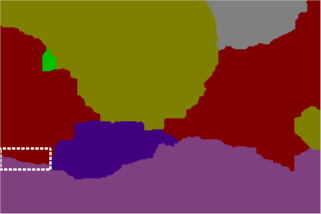}}\\
	\subfloat[The parallel BK-algorithm, 4 threads]{\label{fig:sg4PBK}\includegraphics[width=0.32\linewidth]{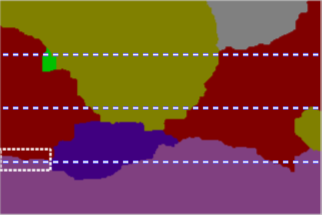}}\,
	\subfloat[The parallel BK-algorithm, 8 threads]{\label{fig:sg8PBK}\includegraphics[width=0.32\linewidth]{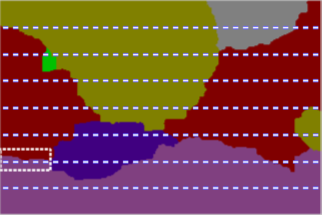}}\,
	\subfloat[The naive converged parallel BK-algorithm, 8 threads]{\label{fig:sgCPBK}\includegraphics[width=0.32\linewidth]{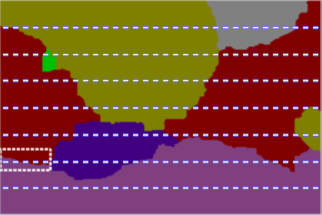}}\\
	\subfloat[Zoom in]{\label{fig:sgDetails}\includegraphics[height = 1.435cm]{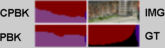}}\quad
	\subfloat[Color-coding legend]{\label{fig:sgLegend}\includegraphics[height = 1.435cm]{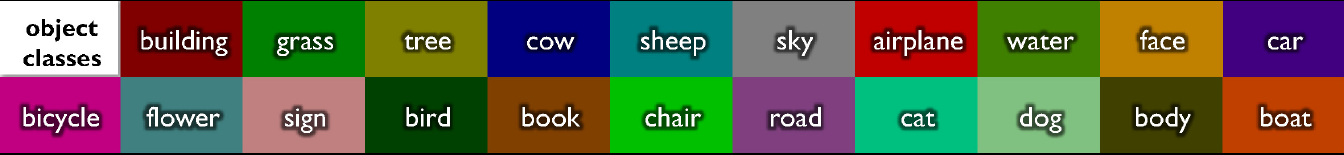}}
    \end{center}
    \caption{Non-converged example of the parallel BK-algorithm}
    \label{fig:failureExample}
\end{figure*}

Since the parallel BK-algorithm is not convergent by only expanding the road class, the differences between the semantic segmentations of the parallel BK-algorithm and those of the BK-algorithm and the naive converged parallel BK-algorithm lie only on the left border between `building' and `road', as shown in the white dotted boxes in Fig.~\ref{fig:failureExample}\subref{fig:sgGT}--\subref{fig:sgCPBK}. For better clarity, these white boxes are enlarged in Fig.~\subref*{fig:sgDetails}, where the result of the naive converged parallel BK-algorithm, which is the same as that of the BK-algorithm, is shown in the upper left. Since the parallel BK-algorithm obtained the same incorrect result under 4 and 8 threads, it is in the bottom left. The image and the ground truth are in the upper right and bottom right, respectively. It is clear that the accuracies, defined as the intersection vs union, of both the building and road classes of the parallel BK-algorithm are worse than those of the naive converged parallel BK-algorithm and the BK-algorithm. Because the road pixels in the white box of the ground truth are all labeled `void', only the building accuracy of the parallel BK-algorithm, which is $85.7\%$, is worse than those of the naive converged parallel BK-algorithm and the BK-algorithm, which are all $86.4\%$. The running time of the BK-algorithm in this example is $0.564$ seconds, while the parallel BK-algorithm takes much longer: $0.998$ seconds with 4 threads and $1.150$ seconds with 8 threads. The time the naive converged parallel BK-algorithm taken is comparable to that of the BK-algorithm: with 4 threads if $\rm ITER$ is set to $15$, $20$, $25$, it takes $0.560$, $0.556$ and $0.571$ seconds, respectively; with 8 threads if $\rm ITER$ is set to $15$, $20$, $25$, it takes $0.613$, $0.610$ and $0.618$ seconds, respectively.

The naive converged parallel BK-algorithm does not have the convergence problem and is much faster than the parallel BK-algorithm on the failure example. However, for the other images, on which the parallel BK-algorithm successfully converged, the naive converged parallel BK-algorithm may unnecessarily invoke merging operations, depending on the setting of $\mathrm{ITER}$, and must pay some extra cost. To evaluate this ``cost'', Fig.\subref*{fig:semanticPGCMerged} shows the relative times of $t_\mathrm{CPBK}/t_\mathrm{PBK}$ on all these images that get converged by the parallel BK-algorithm but unnecessarily invoked merging operations by the naive converged parallel BK-algorithm with $\mathrm{TRDS}$ set to $8$ and $12$ and $\mathrm{ITER}$ set to $15$ and $25$.

It can be seen from Fig. \subref*{fig:semanticPGCMerged} that for both $\mathrm{ITER}$ settings, the running times of the naive converged parallel BK-algorithm with $8$ computational threads are roughly equal to those of the parallel BK-algorithm on average but are slightly shorter than those of the parallel BK-algorithm with $12$ computational threads. This experiment shows that even in the most challenging case, the naive converged parallel BK-algorithm could have the convergence guarantee with almost no additional time cost over the parallel BK-algorithm.

\subsection{Deployment on distributed platforms}
The dynamic parallel graph cuts algorithm (Algorithm \ref{alg:dynamicParallel}) can be deployed on distributed platforms, where both splitting and merging can take place within the same machine or across different machines. Fig. \ref{fig:distCase} gives such an example on an MPICH2 cluster, where the yellow stripes are the overlapping regions within the same machine and the red stripes are the overlapping regions between two machines.
\begin{figure*}[htbp]
	\begin{center}
		\subfloat[Before merging]{\label{fig:distBefore}\includegraphics[width=0.45\linewidth]{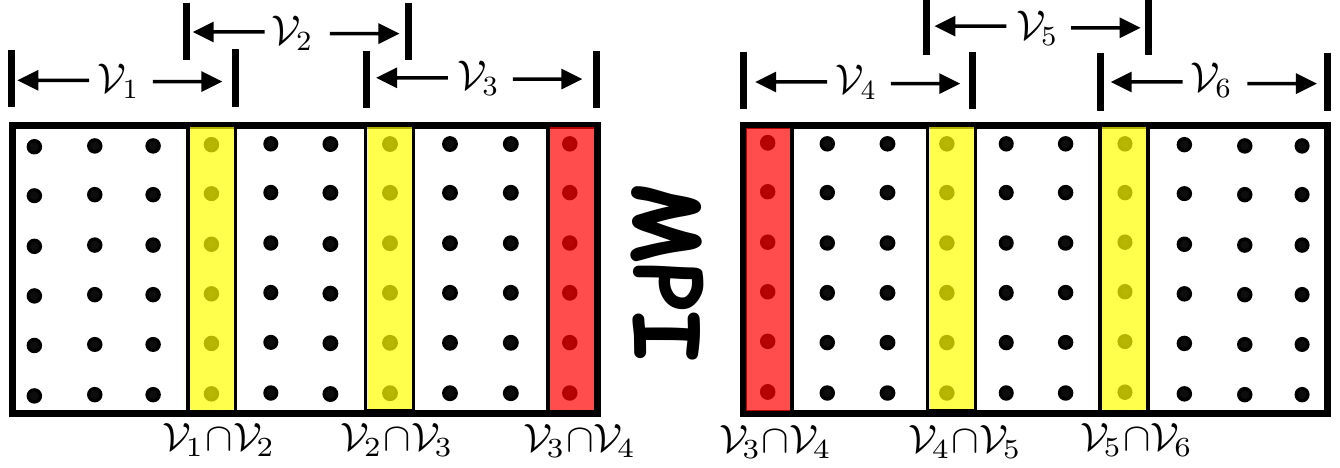}}\qquad\qquad
		\subfloat[After merging]{\label{fig:distAfter}\includegraphics[width=0.45\linewidth]{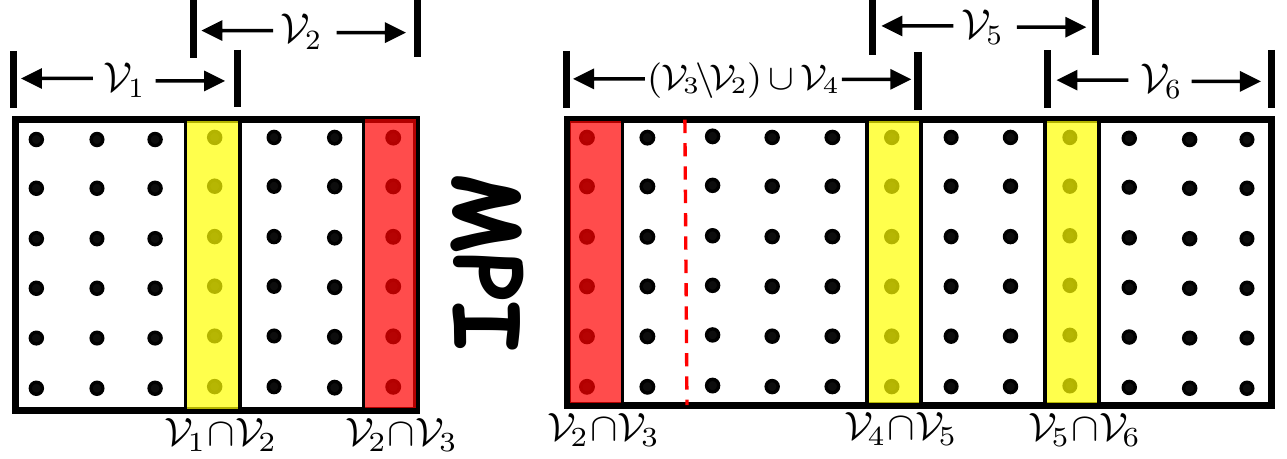}}
	\end{center}
	\caption{Dynamic parallel graph cuts algorithm on an MPICH2 cluster}
	\label{fig:distCase}
\end{figure*}

There is no substantial difference if the merging operations only take place within the same machine. There are two ways to merge the subgraphs that spread across different machines, depending on whether the merged graph can fit in the memory of one of the machines. If it can, merging acts in the same way as on the shared memory platforms, except that MPI~\cite{SnirMPIReference} is used for the inter-machine communication. Fig. \ref{fig:distCase} illustrates such a case, where the rightmost subgraph in the left machine is merged into the machine on the right side. The communication overheads in this case are significant and largely depend on the size of the subgraph being merged. If the merged graph cannot fit in the memory of a single machine, the current splitting can only be ``adjusted'', not simply merged. ``Adjusting'' is achieved via first merging a small portion with the overlapped region and then re-splitting on the new boundary. Therefore only a small part rather than the whole subgraph is merged in this case. It is worth noting that Theorem \ref{thm:conclusion} holds in all cases, such that the merged graph is a reparameterization of the original graph with less flows. Also, note that besides the parallel BK-algorithm, any other distributed graph cuts algorithms could also be used here.

We use some of the big instances in~\cite{maxflowInstances} to evaluate our naive converged parallel BK-algorithm on a $2\times2$ grid of an MPICH2 cluster. The whole graph is split across the four machines and further split into two subgraphs within each machine. In contrast to all the above experiments, $\rm ITER$ has a great impact on the performance of the naive converged parallel BK-algorithm in this distributed setting. The LB07-bunny-lrg dataset, which was solved in $6.3$ seconds by the parallel BK-algorithm, was solved in $6.4$, $6.4$ and $6.6$ seconds by the naive converged parallel BK-algorithm when $\rm ITER$ was set to $25$, $20$ and $15$, respectively. However, further reducing $\rm ITER$ would dramatically increase the running time. The naive converged parallel BK-algorithm took $25.5$ seconds when $\rm ITER=10$, whereas it took $70$ seconds when $\rm ITER=5$. The underlying reason is that network speed is much slower than RAM speed, such that merging subgraphs across different machines is a time-consuming operation. When $\rm ITER$ is relatively small, subgraphs across different machines will be merged. For instance, in the case of $\rm ITER=5$, merging operation was invoked three times such that all eight subgraphs across four machines were merged into one graph, and communication took more than $40$ out of a total $70$ seconds. Thus the naive converged parallel BK-algorithm may pay some higher running time cost for the guaranteed convergence under the above simple merging and splitting strategies.

\section{Conclusion and further work}
To remedy the non-convergence problem in the parallel BK-algorithm, we first introduced a merging method along with the naive converged parallel BK-algorithm that simply merges every two neighboring subgraphs once the algorithm is found to have difficulty converging under the current subgraph splitting configuration. We then introduced a new pseudo-boolean representation for graph cuts, namely the restricted homogeneous posiforms. We further developed an invariance analysis method for graph cuts algorithms and proved the correctness and effectiveness of our proposed merging method. Based on the merging method, we proposed a general parallelization framework called the dynamic parallel graph cuts algorithm, which allows the subgraph splitting configuration to be adjusted freely during the parallelization process and has guaranteed convergence at the same time. Extensive experiments for image segmentation and semantic segmentation with different number of computational threads demonstrated that even under the simplest merging strategy, our proposed naive converged parallel BK-algorithm not only has guaranteed convergence but also has competitive computational efficiency compared with the parallel BK-algorithm.

How to split a graph for parallel computation and how to dynamically adjust the subgraph splitting configuration during the parallelization process in order to further boost the efficiency and accelerate the convergence rate are our future research directions.

\section*{Acknowledgment}
This work was supported by National Natural Science Foundation of China under grants 61333015, 61421004, 61473292. We also thank Dr. Nilanjan Ray and the anonymous reviewers for their valuable comments.

\bibliographystyle{IEEEtran}
\bibliography{Reference}

\appendices
\counterwithin{figure}{section}
\section{Proof of Proposition \ref{prop:split}} \label{app:split}
\begin{proof}
	Fig.\ref{fig:graphSplitting}\subref{fig:mergedGraph}--\subref{fig:splitGraphR} give an example of splitting, where the capacities of the edges in the two split subgraphs, whose two terminals are common for the two subgraphs, are weighted by $1/2$ and all the other capacities simply come from those of their corresponding edges in the original graph. Therefore, the following relations hold for the coefficients of the restricted homogeneous posiforms of the two split subgraphs $G_1^{(1)}$, $G_2^{(1)}$ and the original graph $G$:
	\begin{equation}
	\widehat{\underaccent{\dot}{a}}_i^{1^{(1)}} \!\!= \widehat{\underaccent{\dot}{a}}_i  ,\,\,\,\,  \widehat{\underaccent{\ddot}{a}}_i^{1^{(1)}} \!\!= \widehat{\underaccent{\ddot}{a}}_i ,\,\,  \forall i\in \mathcal{V}_1\!\!\setminus\!\!\mathcal{V}_2,
	\end{equation}
	\begin{equation}
	\widetriangle{\underaccent{\dot}{a}}_i^{2^{(1)}} \!\!= \widetriangle{\underaccent{\dot}{a}}_i  ,\,\,\,\,  \widetriangle{\underaccent{\ddot}{a}}_i^{2^{(1)}} \!\!= \widetriangle{\underaccent{\ddot}{a}}_i ,\,\,  \forall i\in \mathcal{V}_2\!\!\setminus\!\!\mathcal{V}_1,
	\end{equation}
	\begin{equation}
	\bar {\underaccent{\dot}{a}}_i^{1^{(1)}}\!\! = \bar {\underaccent{\dot}{a}}_i^{2^{(1)}}\!\!=\frac{1}{2}\bar {\underaccent{\dot}{a}}_i,\,\,\,\,\bar {\underaccent{\ddot}{a}}_i^{1^{(1)}}\!\! = \bar {\underaccent{\ddot}{a}}_i^{2^{(1)}} \!\!= \frac{1}{2}\bar {\underaccent{\ddot}{a}}_i,\,\, \forall i\in \mathcal{V}_1\!\cap\!\mathcal{V}_2,
	\end{equation}
	\begin{equation}
	\widehat a_{ij}^{1^{(1)}}\!\! = \widehat a_{ij},\forall i\!\text{ \it or }\!j \in \mathcal{V}_1\setminus\mathcal{V}_2,
	\end{equation}
	\begin{equation}
	\widetriangle a_{ij}^{2^{(1)}}\!\! = \widetriangle a_{ij},\forall i\!\text{ \it or }\!j \in \mathcal{V}_2\setminus\mathcal{V}_1,
	\end{equation}
	\begin{equation}
	\bar a_{ij}^{1^{(1)}} \!\!= \bar a_{ij}^{2^{(1)}} \!\!=\frac{1}{2}\bar a_{ij},\,\,\forall i,j\in \mathcal{V}_1\!\cap\!\mathcal{V}_2.
	\end{equation}
	And from the relations of the coefficients of the restricted homogeneous posiforms and the multi-linear polynomials, stated in (\ref{eq:relationFPhi^})--(\ref{eq:relationF2Phi2$}), it is easy to verify that all the relations listed in this theorem hold.
\end{proof}

\section{Proof of Proposition \ref{prop:maxflow}} \label{app:maxflow}
\begin{proof}
All the flows are pushed through the augmenting paths in the augmenting path maxflow algorithms. Without loss of generality, suppose $s \rightarrow i_1 \rightarrow i_2 \rightarrow \cdots \rightarrow i_k \rightarrow t $ is such an augmenting path in graph $G$ with the minimal capacity of forward arc (edge) greater than or equal to $\mathrm{F}$. Therefore a flow of $\mathrm{F}$ can be pushed along this augmenting path, resulting in a residual graph $G'$. It is clear that graph $G$ and residual graph $G'$ only differ in these forward and backward arcs, such that the restricted homogeneous posiforms of the two graphs can be written as:
\begin{equation}
\phi_h^G(\mathbf{x}) = \accentset{a}{\phi}^G_h(\mathbf{x}) + \accentset{r}{\phi}^G_h(\mathbf{x}),
\end{equation}
\begin{equation}
\phi_h^{G'}(\mathbf{x}) = \accentset{a}{\phi}^{G'}_h(\mathbf{x}) + \accentset{r}{\phi}^{G'}_h(\mathbf{x}),
\end{equation}
where $\accentset{a}{\phi}^G_h(\mathbf{x})$ and $\accentset{a}{\phi}^{G'}_h(\mathbf{x})$ contain all the components corresponding to all the forward and backward arcs along the augmenting path of graph $G$ and $G'$ respectively. $\accentset{r}{\phi}^G_h(\mathbf{x})$ and $\accentset{r}{\phi}^{G'}_h(\mathbf{x})$ contain all the rest components of graph $G$ and $G'$, such that $\accentset{r}{\phi}^{G}_h(\mathbf{x}) = \accentset{r}{\phi}^{G'}_h(\mathbf{x})$. Suppose that $\accentset{a}{\phi}^G_h(\mathbf{x})$ is in the following form:
\begin{equation}
\accentset{a}{\phi}^G_h(\mathbf{x}) = 
a_{i1}x_{i1} + a_{i_1i_2}\bar x_{i_1}x_{i_2} + a_{i_2i_1}\bar x_{i_2}x_{i_1} + \cdots + a_{i_k}\bar x_{k},
\end{equation}
where all the coefficients of the forward arcs (\eg $a_{i_1i_2}$) are greater than or equal to $\mathrm{F}$. Pushing a flow $\mathrm{F}$ in this augmenting path will cause all the capacities of the forward arcs reduced by $\mathrm{F}$ and all the capacities of the backward arcs increased by $\mathrm{F}$. Therefore $\accentset{a}{\phi}^{G'}_h(\mathbf{x})$ is in the following form
\begin{equation}
\begin{split}
\accentset{a}{\phi}^{G'}_h(\mathbf{x}) & = (a_{i1} - \mathrm{F})x_{i1} + (a_{i_1i_2} - \mathrm{F})\bar x_{i_1}x_{i_2} + \\
&\mathrel{\phantom{=}} (a_{i_2i_1} + \mathrm{F}) \bar x_{i_2}x_{i_1} 
+  \cdots + (a_{i_k}- \mathrm{F})\bar x_k.
\end{split}
\end{equation}
Therefore,
\begin{equation}
\begin{split}
&{}\accentset{a}{\phi}^G_h(\mathbf{x}) - \accentset{a}{\phi}^{G'}_h(\mathbf{x}) =\\
&{}\mathrm{F}\times\left(
x_{i_1} + \bar x_{i_1}x_{i_2} -\bar x_{i_2}x_{i_1} + \cdots + \bar x_{k}
\right),
\end{split}
\end{equation}
it is easy to verify $x_{i_1} + \bar x_{i_1}x_{i_2} -\bar x_{i_2}x_{i_1} + \cdots + \bar x_{k} \equiv 1 $, hence 
\begin{equation}
\accentset{a}{\phi}^G_h(\mathbf{x}) - \accentset{a}{\phi}^{G'}_h(\mathbf{x})  \equiv  \mathrm{F}.
\end{equation}
Or $\phi^G_h(\mathbf{x})  = \phi^{G'}_h(\mathbf{x})  + \mathrm{F}$, and Proposition \ref{prop:maxflow} is proved.
\end{proof}

\section{Proof of Proposition \ref{prop:update}} \label{app:update}
\begin{proof}
Updating operations only modify the T-link capacities of all the nodes within $\mathcal{V}_1\!\cap\!\mathcal{V}_2$ in the following way:
\begin{equation}
\bar{\underaccent{\dot}{a}}^{1^{(k+1)}}_i = \bar{\underaccent{\dot}{a}}^{1^{(k)}}_i + \triangle\lambda_i^{(k+1)}, \forall i\in\mathcal{V}_1\!\cap\!\mathcal{V}_2,
\end{equation}
\begin{equation}
\bar{\underaccent{\ddot}{a}}^{2^{(k+1)}}_i = \bar{\underaccent{\ddot}{a}}^{2^{(k)}}_i + \triangle\lambda_i^{(k+1)}, \forall i\in\mathcal{V}_1\!\cap\!\mathcal{V}_2.
\end{equation}
And from (\ref{eq:relationF1Phi1barNode}) and (\ref{eq:relationF2Phi2barNode}), it is easy to verify that all the conclusions in this proposition hold.
\end{proof}

\section{Proof of Proposition \ref{prop:merge}} \label{app:merge}
\begin{proof}
    From the merging operation described in section \ref{sec:mergingMethod} and Theorem \ref{thm:homogeneousRep}, the following relations hold for the coefficients of the restricted homogeneous posiforms of the merged graph $G'$ and those of the two subgraphs $G_1^{(k+1)}$, $G_2^{(k+1)}$
\begin{equation}
\widehat{\underaccent{\dot}{a}}'_i   = \widehat{\underaccent{\dot}{a}}_i^{1^{(k+1)}} \!\! ,\,\,\,\,  \widehat{\underaccent{\ddot}{a}}'_i = \widehat{\underaccent{\ddot}{a}}_i^{1^{(k+1)}} \!\!  ,\,\,  \forall i\in \mathcal{V}_1\!\!\setminus\!\!\mathcal{V}_2
\end{equation}
\begin{equation}
\widetriangle{\underaccent{\dot}{a}}'_i= \widetriangle{\underaccent{\dot}{a}}_i^{2^{(k+1)}} \!\!  ,\,\,\,\,   \widetriangle{\underaccent{\ddot}{a}}'_i = \widetriangle{\underaccent{\ddot}{a}}_i^{2^{(k+1)}} \!\! ,\,\,  \forall i\in \mathcal{V}_2\!\!\setminus\!\!\mathcal{V}_1
\end{equation}
\begin{equation}
    \begin{split}
\bar {\underaccent{\dot}{a}}'_i &=\bar {\underaccent{\dot}{a}}_i^{1^{(k+1)}}\!\!+ \bar {\underaccent{\dot}{a}}_i^{2^{(k+1)}}\\
\bar {\underaccent{\ddot}{a}}'_i &= \bar {\underaccent{\ddot}{a}}_i^{2^{(k+1)}}\!\! + \bar {\underaccent{\ddot}{a}}_i^{1^{(k+1)}} 
\end{split}
\,\,\,,\,\quad\forall i\in \mathcal{V}_1\!\cap\!\mathcal{V}_2
\end{equation}
\begin{equation}
\widehat a'_{ij} = \widehat a_{ij}^{1^{(k+1)}} ,\,\,\forall i\!\text{ \it or }\!j \in \mathcal{V}_1\setminus\mathcal{V}_2
\end{equation}
\begin{equation}
\widetriangle a'_{ij} = \widetriangle a_{ij}^{2^{(k+1)}},\,\,\forall i\!\text{ \it or }\!j \in \mathcal{V}_2\setminus\mathcal{V}_1
\end{equation}
\begin{equation}
\bar a'_{ij} = \bar a_{ij}^{1^{(k+1)}} \!\! + \bar a_{ij}^{2^{(k+1)}}  ,\,\,\forall i,j\in \mathcal{V}_1\!\cap\!\mathcal{V}_2
\end{equation}
And all the four equalities in Proposition \ref{prop:merge} can be verified by combining the relations of the coefficients of the restricted homogeneous posiforms and the multi-linear polynomials, stated in (\ref{eq:relationFPhi^})--(\ref{eq:relationF2Phi2$}).
\end{proof}

\section{An Illustrative Example of Equivalent Graph Cuts Problems} \label{app:illustrative}
Here we give an illustrative example on the equivalence of graph cuts problems. For the three graph cuts problems defined on the graphs shown in Fig.~\ref{fig:pushFlows}\subref{fig:G0}-\subref{fig:G2} respectively, their unique restricted homogeneous posiforms, according to Theorem \ref{thm:homogeneousRep}, are 
\begin{figure}[htbp]
	\begin{center}
		\subfloat[Original graph $G_0$]{\label{fig:G0}\includegraphics[width=0.3\linewidth]{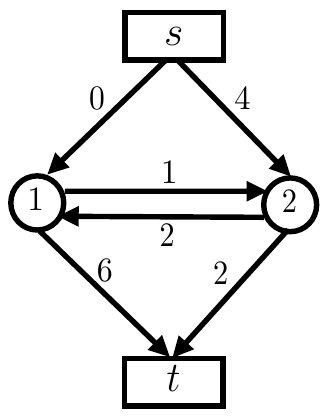}}\quad
		\subfloat[Residual graph $G_1$]{\label{fig:G1}\includegraphics[width=0.3\linewidth]{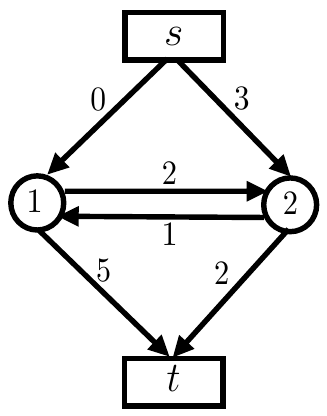}}\quad
		\subfloat[Residual graph $G_2$]{\label{fig:G2}\includegraphics[width=0.3\linewidth]{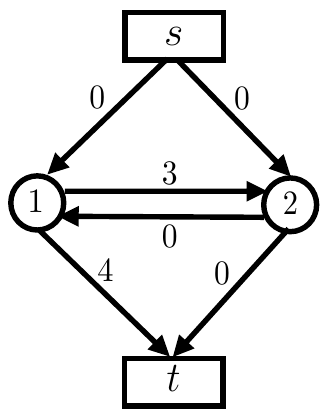}}		
	\end{center}
	\caption{Three equivalent graph cuts problems.}
	\label{fig:pushFlows}
\end{figure}

\begin{align}
\phi^{G_0}_h(\mathbf{x}_\mathcal{V}) & = 6\bar x_1 + 4x_2 + 2\bar x_2 + \bar x_1x_2 + 2\bar x_2x_1\,,\label{eq:hG0}\\
\phi^{G_1}_h(\mathbf{x}_\mathcal{V}) & = 5\bar x_1 + 3x_2 + 2\bar x_2 + 2\bar x_1x_2 + \bar x_2x_1\,, \label{eq:hG1}\\
\phi^{G_2}_h(\mathbf{x}_\mathcal{V}) & = 4\bar x_1 + 3\bar x_1x_2\,.  \label{eq:hG2}
\end{align}
And their corresponding multi-linear polynomials are
\begin{align}
f^{G_0}(\mathbf{x}_\mathcal{V}) & = 8 - 4x_1 + 3x_2 - 3x_1x_2 \,,\label{eq:fG0}\\
f^{G_1}(\mathbf{x}_\mathcal{V}) & = 7 - 4x_1 + 3x_2 - 3x_1x_2 \,, \label{eq:fG1}\\
f^{G_2}(\mathbf{x}_\mathcal{V}) & = 4 - 4x_1 + 3x_2 - 3x_1x_2 \,.  \label{eq:fG2}
\end{align}
Obviously, 
\begin{equation}
f^{G_0}(\mathbf{x}_\mathcal{V}) = f^{G_1}(\mathbf{x}_\mathcal{V}) + 1 = f^{G_2}(\mathbf{x}_\mathcal{V}) + 4.
\end{equation}

Therefore, $f^{G_0}(\mathbf{x}_\mathcal{V})$ and $f^{G_1}(\mathbf{x}_\mathcal{V})$ only differ by a constant of 1, and $f^{G_0}(\mathbf{x}_\mathcal{V})$ and $f^{G_2}(\mathbf{x}_\mathcal{V})$ only differ by a constant of 4. Note that ``1'' equals the total flow pushed through the path $s\rightarrow 2\rightarrow 1\rightarrow t$ in $G_0$, by which the residual graph $G_1$ is produced, and ``4'' equals the sum of the total flows pushed through the path $s\rightarrow 2\rightarrow 1\rightarrow t$, and those through the path $s\rightarrow 2\rightarrow t$ in $G_0$, by which the residual graph $G_2$ is produced. Proposition \ref{prop:maxflow} showed that pushing a flow $\mathrm{F}$ in graph G will only reduce the multi-linear polynomial of $G$ by a constant of $\mathrm{F}$. Here, the flow $\mathrm{F}$ varies by each augmenting path operation, but it is independent of the variables of the pseudo-boolean function representing the graph cuts problem. Hence the three graph cuts problems are equivalent.

\begin{biography}[{\includegraphics[width=1in,clip,keepaspectratio]{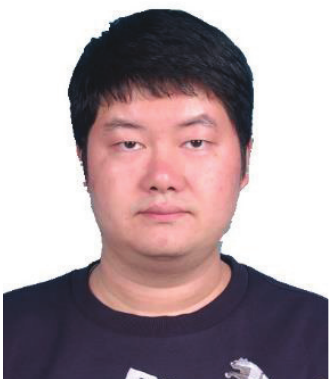}}]{Miao Yu} received his B.B.A. and M.S. both from Southwest Jiaotong University in 2004 and 2007 respectively, and the Ph.D. from the Institute of Automation, Chinese Academy of Sciences in 2016. Since 2007, he has been with the Zhongyuan University of Technology, where he is now a lecturer. His research interests include parallel computing, probabilistic graphical models, higher order energy minimization and scene understanding.
\end{biography}

\begin{biography}[{\includegraphics[width=1in,clip,keepaspectratio]{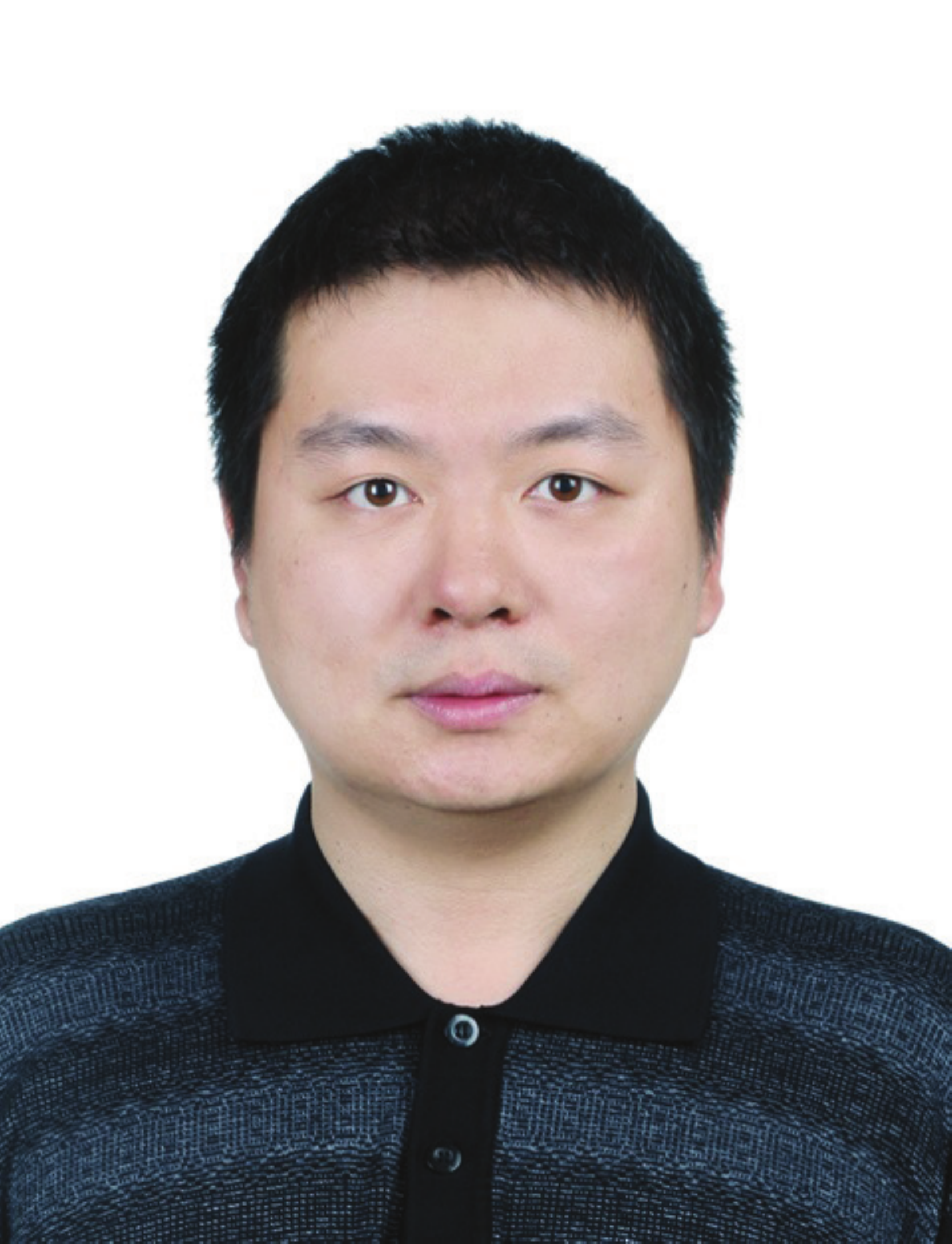}}]{Shuhan Shen} received his B.S. and M.S. both from Southwest Jiao Tong University in 2003 and 2006 respectively, and the Ph.D. from Shanghai Jiao Tong University in 2010. Since 2010, he has been with the National Laboratory of Pattern Recognition at Institute of Automation, Chinese Academy of Sciences, where he is now an associate professor. His research interests are in 3D computer vision, which include image based 3D modeling of large scale scenes, 3D perception for intelligent robot, and 3D semantic reconstruction.
\end{biography}

\begin{biography}[{\includegraphics[width=1in,clip,keepaspectratio]{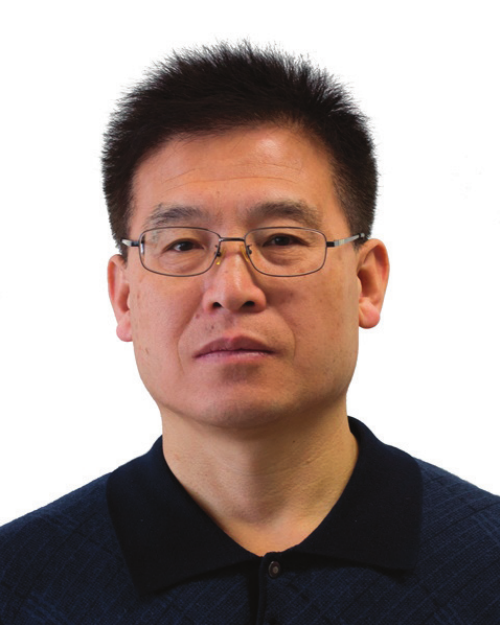}}]{Zhanyi Hu} received his B.S. Degree in Automation from the North China University of Technology in 1985, the Ph.D. Degree (Docteur d'\'{E}tat) in Computer Science from the University of Liege, Belgium, in Jan. 1993. Since 1993, he has been with the Institute of Automation, Chinese Academy of Sciences. From 1997 to 1998, he was a Visiting Scientist in the Chinese University of Hong Kong. From 2001 to 2005, he was an executive panel member of National High-Tech R\& D Program (863 Program). From 2005 to 2010, he was a member of the Advisory Committee of the National Natural Science Foundation of China. Dr. Hu is now a research professor of computer vision, a deputy editor-in-chief for Chinese Journal of CAD\& CG, and an associate editor for Science in China, and Journal of Computer Science and Technology. He was the organization committee co-chair of ICCV2005, and a program co-chair of ACCV2012. Dr. Hu has published more than 150 peer-reviewed journal papers, including IEEE PAMI, IEEE IP, IJCV. His current research interests include biology-inspired vision and large scale 3D reconstruction from images.
\end{biography}
\end{document}